\def\dfinitynote{Part of this work was done while at DFINITY.}
\documentclass[sigconf,screen]{acmart}
\usepackage{algorithm}
\usepackage[noend]{algpseudocode}
\usepackage{array}
\usepackage{circuitikz}
\usepackage{mathtools}
\usepackage{nicematrix}
\usepackage{subcaption}
\usepackage{tikz}
\usepackage{tikzit}
\usepackage{wrapfig}
\algdef{SE}[EVENT]{Upon}{EndUpon}[1]{\textbf{upon}\ #1\ \algorithmicdo}{\algorithmicend~\textbf{upon}}
\makeatletter
\ifthenelse{\equal{\ALG@noend}{t}}%
  {\algtext*{EndUpon}}
  {}%
\makeatother%
\tikzstyle{red}=[draw=none, fill=red, shape=circle, tikzit shape=circle]

\tikzstyle{data}=[draw=black, thick, ->, >=latex]
\tikzstyle{control}=[draw={rgb,255: red,64; green,64; blue,64}, densely dashed, ->, >=Latex]
\tikzstyle{dotted}=[-, draw={rgb,255: red,191; green,191; blue,191}]
\tikzstyle{data fail}=[-, draw=black, thick]
\tikzstyle{cross}=[-, draw=black, thick]

\usepackage{nicematrix}
\usepackage{tikz}
\usetikzlibrary{decorations.pathreplacing}
\begin{document}
\title[AVID with Low Space and Communication Complexity]{ Asynchronous Verifiable Information Dispersal\texorpdfstring{\\}{}with
%Efficient Dispersal, Storage, Retrieval, and Recovery}
Low Space and Communication Complexity}

\author{Thomas Locher}
\affiliation{%
  \institution{DFINITY}
  \country{Switzerland}
}

\author{Yvonne-Anne Pignolet}
\ifdefined\dfinitynote\authornote{\dfinitynote}\fi
\affiliation{%
  \institution{StableClear}
  \country{Switzerland}
}


% TODO: add concept_ids below from https://dl.acm.org/ccs (needed for HotCRP)


\copyrightyear{2026}
\acmYear{2026}
\setcopyright{cc}
\setcctype{by}
\acmConference[SPAA '26]{38th ACM Symposium on Parallelism in Algorithms and Architectures}{July 06--10, 2026}{London, United Kingdom}
\acmBooktitle{38th ACM Symposium on Parallelism in Algorithms and Architectures (SPAA '26), July 06--10, 2026, London, United Kingdom}
\acmDOI{10.1145/3816782.3819202}
\acmISBN{979-8-4007-2761-0/2026/07}

% !TEX root = avid_spaa.tex
\begin{abstract}
The primary goal of a distributed storage system is to ensure that clients can both write and read data in a reliable and consistent manner, even in the presence of failures. While existing asynchronous verifiable information dispersal (AVID) protocols achieve optimal space complexity for storage and communication complexity for data retrieval in a Byzantine setting, the crucial operations of data dispersal and node recovery have received less attention.
We propose an efficient AVID protocol that simultaneously guarantees low complexities for dispersal, storage, retrieval, and recovery. At the core of the proposed protocol lies a novel mechanism to encode data in a two-dimensional matrix and a bespoke dispersal algorithm. The protocol maintains an optimal communication complexity for retrieval while substantially improving upon the state of the art for recovery.
Additionally, we describe a protocol variant that offers a reduced space complexity and communication complexity for dispersal, at the expense of a higher communication complexity for recovery and, depending on its parameterization, also retrieval.
As the proposed protocols strike a balance across all considered metrics, they are suitable for a broad range of real-world use cases.
\end{abstract}

\maketitle

% !TEX root = ./avid_spaa.tex

\section{Introduction}
\label{sec:introduction}

Distributed storage systems are a foundational pillar of modern infrastructure, supporting applications that range from global-scale cloud services to decentralized archives and blockchain scalability solutions. Emerging decentralized storage platforms such as Walrus~\cite{danezis2025} and Shelby~\cite{goren2025} are beginning to deploy these techniques at scale, underscoring the practical relevance of distributed storage systems.
The primary bottleneck for scaling such systems is the significant overhead---comprising storage footprint, bandwidth consumption, and maintenance costs---required to manage data across nodes that may (and do) fail at any time.
Offering resilience against a broad range of disruptions, including network partitions, software bugs, and malicious behavior, is a challenging requirement.
\emph{Asynchronous verifiable information dispersal} (AVID) protocols have been established as a core primitive to provide robust guarantees without any assumptions on network delays.

Since their introduction by Cachin and Tessaro~\cite{cachin2005}, AVID protocols have been evaluated primarily on two metrics: space complexity (storage requirement) and communication complexity (bandwidth requirement) for data retrieval.
State-of-the-art constructions achieve a \emph{storage overhead factor} of $1.5$ (i.e., $1.5$ times the original data is stored) when strictly less than a third of all nodes are \emph{Byzantine}, that is, exhibit arbitrary faulty behavior. This factor is optimal: for any data item $m$, storing less than $1.5|m|$ bits in total would allow $\approx$$n/3$ Byzantine nodes to concentrate their share and deny access to $m$, where $|m|$ denotes the size of $m$ in bits and $n$ is the total number of storage nodes. 
This focus, however, does not reflect the operational realities of many modern workloads.
In particular, there is an explosive growth of ``dark data'' and \emph{write-once-read-never} archives~\cite{splunk2019, aaa2025}, estimated to account for 30--80\% of all data in highly regulated industries such as healthcare and finance, where information must be reliably ingested and retained over long periods of time but is practically never accessed.
Thus, the costs of initial data dispersal (distributing data across nodes), continuous storage footprint, background node repair, and retrieval (accessing data on demand) must all be considered, and the right balance among them depends on the use case scenario.

Many AVID protocols incur a significant dispersal overhead and lack efficient \emph{node recovery}: when a node fails, the replacement node has to download the entire dataset to reconstruct its share of the data---an approach that is practically infeasible at scale, even when node failures are infrequent.
Recent work has begun to address this shortcoming by reducing the dispersal and recovery overheads.
For example, Alhaddad et al.~\cite{alhaddad2022} proposed an AVID protocol with an overhead factor for dispersal of only $3$; however, the protocol incurs the same (sub-optimal) overhead factor of $3$ for storage and data retrieval as well.
It was later shown how to keep the overhead factor of storage and data retrieval at $1.5$ at a higher but still constant overhead factor for data dispersal~\cite{alhaddad2024}.  This approach is compatible with minimum storage regenerating (MSR) codes~\cite{dimakis2010} to improve recovery performance. While these codes enable efficient recovery, they typically incur higher communication costs during retrieval.
Danezis et al.~\cite{danezis2025} introduced the Red Stuff protocol, which prioritizes efficient recovery. % through a specialized encoding.
Yet, this efficiency comes at a steep cost, requiring a storage overhead factor of $4.5$, three times the theoretical minimum.

In this paper, we present AVID protocols that improve upon the state of the art across the entire operational spectrum.
Our approach is built from standard primitives (erasure codes, vector commitments); the contribution lies in a novel two-dimensional encoding scheme and dispersal algorithm.
Concretely, our main contribution is a protocol that achieves a storage overhead factor of $3$, a 33\% improvement over Red Stuff, while simultaneously reducing communication complexity for both dispersal and node recovery by 14\% and 40\%, respectively.
As an additional contribution, we introduce a parameterized variant of our protocol based on MSR codes, trading off the cost of storage and dispersal for retrieval and recovery efficiency, thereby enabling deployment tuning to use-case requirements.

Table~\ref{tab:comparison} provides a comparison of various distributed storage protocols.
The additive term independent of $|m|$ is $O(n^2|\pi|)$ for all metrics across all listed protocols, where $|\pi|$ denotes the commitment-proof size, which is discussed in detail in \S\ref{sec:model}. Depending on the commitment and batching scheme, the additive term for storage and retrieval is $O(n|\pi| + n^2\log(n))$ for certain protocols.
While no single protocol dominates all metrics, the proposed protocols show that efficient recovery can be achieved with lower storage overhead and reduced dispersal cost.

\begin{table}[t]
  \begin{center}
    \begin{tabular}{ l r r r r}
      Protocol 								& Storage 	& Dispersal 	& Retrieval 	& Recovery \\
      \hline
      Cachin-Tessaro~\cite{cachin2005} & $1.5 |m| $	& $\boldsymbol{ n |m|}$		& $1.5 |m|$	& $\boldsymbol{1.5 |m|}$ \\
      Alhaddad et al.~\cite{alhaddad2022}	& $3 |m| $	& $3 |m| $		& $3 |m|$	& $\boldsymbol{3 |m|}$ \\
      Alhaddad et al.~\cite{alhaddad2024} 	& $1.5 |m| $	& $9 |m| $		& $1.5 |m|$	& $\boldsymbol{1.5 |m|}$ \\
      Danezis et al.~\cite{danezis2025}			& $4.5 |m| $	& $7 |m| $		& $1.5 |m|$	& $7.5 {|m|}/{n}$ \\
      This work (\S\ref{sec:avid})				 & $3 |m| $	& $6 |m| $	& $1.5 |m|$	& $4.5 {|m|}/{n}$ \\
      This work (\S\ref{sec:avid_msr}, $\gamma = 1.8$) 			& $1.8 |m| $	& $6.12 |m| $		& $1.6 |m|$	& $16.2 {|m|}/{n}$ \\
      This work (\S\ref{sec:avid_msr}, $\gamma = 1.5$) 			& $1.5 |m| $	& $4.5 |m| $		& $1.5 |m|$	& $\boldsymbol{1.5 |m|}$ \\
    \end{tabular}
    \caption{Comparison of the space and communication complexity for distributed storage protocols for data item $m$ of size $|m|$ with up to $\lfloor (n-1)/3 \rfloor$ Byzantine nodes in a network of $n$ nodes. Entries in \textbf{bold} are \emph{not} asymptotically optimal with respect to $|m|$ and $n$.}~\label{tab:comparison}
  \end{center}\vspace{-14pt}
\end{table}

%Note that depending on the instantiation of the commitment scheme used in these protocols, the storage and communication overhead as well as cryptographic assumptions and setup requirements vary (but remain polynomial in~$n$).
%Since the commitment proof overhead is negligible compared to the data size for large data, only the overhead with respect to the data size is shown in this table.\yap{this paragraph could be omitted}

Note that the main protocol (\S\ref{sec:avid}) guarantees an optimal overhead for data retrieval and strictly outperforms the protocol of Danezis et al.\ in all other metrics.
The complexities of the protocol variant (\S\ref{sec:avid_msr}) can be tuned using a parameter $\gamma \ge 3/2$, which corresponds to the storage overhead factor.
When choosing, e.g., $\gamma = 1.8$, the space complexity is nearly optimal
but it still underperforms the main protocol in all other dimensions.
For $\gamma = 3/2$, the complexities match the complexities of the protocol by Alhaddad et al.~\cite{alhaddad2024}, with the dispersal complexity cut in half.
Thus, we also improve upon the state of the art for the class of space-optimal AVID protocols.

The constants in Table~\ref{tab:comparison} translate directly to operating costs.
Consider a network of $n=100$ nodes with 1~Gbps bandwidth per node storing a 1~GB data item.
When using our main protocol, each node stores 30~MB, versus 45~MB when using Red~Stuff---a saving of 15~MB per item per node.
The commitment-proof overhead is $\approx$60~KB (using Merkle trees) or $\approx$6~KB (using KZG commitments~\cite{kate2010}) per item per node---three to four orders of magnitude below the data footprint---, confirming that the term dependent on $|m|$ dominates operational costs.
At corpus scale (1~PB of content), Red~Stuff requires 4.5~PB of network-wide storage while our main protocol requires only 3~PB, a reduction of 1.5~PB.
When a node fails, the replacement node must download 75~TB under Red~Stuff versus 45~TB under our protocol in  \S\ref{sec:avid}, saving 30~TB ($\approx$2.8~days of repair traffic at saturated bandwidth).
Depending on storage and bandwidth costs, the protocol variant in \S\ref{sec:avid_msr} is attractive for deployments that prioritize minimizing storage further, accepting a higher recovery and, for $\gamma > 3/2$, retrieval overhead.

The remainder of the paper is organized as follows. The system model and preliminaries are introduced in~\S\ref{sec:model}. The main AVID protocol and the protocol variant are presented in~\S\ref{sec:avid} and \S\ref{sec:avid_msr}, respectively. Related work is summarized in~\S\ref{sec:related_work}, and~\S\ref{sec:conclusion} concludes the paper.

% !TEX root = ./avid_spaa.tex
\section{Model and Preliminaries}
\label{sec:model}

The considered system consists of a storage network of $n$ nodes $\{v_1, \ldots , v_n\}$ and an unbounded set of clients $\{c, c', \ldots \}$. The nodes comprising the storage network are also called \emph{storage nodes}.
The goal is to store data efficiently and keep the stored data consistent and available despite the presence of up to $t \coloneqq \lfloor (n-1)/3\rfloor$ Byzantine nodes in the storage network. A Byzantine node may deviate from a given protocol in arbitrary ways, e.g., it may not store data, provide invalid or inconsistent data, or simply crash and stop participating in protocol execution entirely.
Clients may also be Byzantine, aiming to transition the storage network into an inconsistent state or causing undue load in terms of consumed storage space and bandwidth, possibly by colluding with the Byzantine nodes.
Nodes and clients that never fail and adhere to a given protocol at all times are called \emph{honest}.

We assume that the network topology is static and known to all storage nodes and clients and that all communication occurs over authenticated channels. Furthermore, communication is assumed to be \emph{asynchronous}, i.e., every message is guaranteed to arrive at the targeted recipient but the delay until it is received is unbounded.
In other words, there is an \emph{adversary} with complete control over message scheduling, yet it cannot drop or alter messages of honest storage nodes and clients, nor create messages on their behalf. The adversary can adaptively corrupt nodes and clients during protocol execution, subject to the constraint that at most $t$ nodes may be corrupted. Any corrupted node is fully controlled by the adversary.

Naturally, it must be possible to store data in the storage network. To this end, a client executes a \emph{dispersal} algorithm to have some data item $m$ of size $|m|$ stored. Throughout this paper, $|m|$ denotes the worst-case (i.e., incompressible) size of $m$ in bits. Once the data item is stored, any client can execute a \emph{retrieval} algorithm to read $m$ from the storage network.
Each data item $m$ is assigned a unique identifier $id =$ \texttt{ID($m$)}, where \texttt{ID} is a function mapping from the message space to an identifier space with a negligible probability that any two data items can be found that map to the same identifier.
As mentioned before, the client may be Byzantine as well. A Byzantine client may disseminate invalid data in the storage network, which is acceptable assuming that the client has paid for the consumed resources beforehand.
If the stored data item is invalid, the result of retrieval with parameter $id$ must be $\bot$, indicating that no valid data item is stored for this identifier.

The conditions of AVID, which have been slightly rephrased from their original specification~\cite{cachin2005}, are captured formally in the following definition.

\begin{definition}[Asynchronous verifiable information dispersal (AVID)]\label{def:avid} In an asynchronous network of $n$ nodes where up to $t < n/3$ nodes may be Byzantine, the following conditions must hold for some parameter $\ell \le n-t$.
\begin{itemize}
 \item \textbf{Termination}: If an honest client $c$ disperses $m$, then all honest nodes eventually terminate the dispersal of $m$.
 \item \textbf{Agreement}: If some honest node terminates the dispersal of $m$, then all honest nodes eventually terminate the dispersal.
 \item \textbf{Availability}: If $\ell$ honest nodes terminate the dispersal of $m$, then any client can eventually retrieve some data item.
 \item \textbf{Correctness}: If $\ell$ honest nodes terminate the dispersal of $m$ and clients $c'$ and $c''$ retrieve $m'$ and $m''$, then $m' = m''$. If the writing client is honest, then $m = m' = m''$.
\end{itemize}
\end{definition}

The parameter $\ell$ is related to the storage overhead as well as the communication complexity. Our protocol in \S\ref{sec:avid} uses $\ell \coloneqq n-t$ as it minimizes the complexities but it can be made to work for smaller $\ell$ as well.
In short, the conditions state that if one honest node terminates the dispersal, then all other honest nodes terminate the dispersal, and all clients read consistent results (as soon as at least $\ell$ nodes have terminated the dispersal). Moreover, if the writing client $c$ is honest, it is further guaranteed that all reading clients read client $c$'s dispersed data item.

AVID protocols typically do not consider the fact that nodes must eventually be replaced in long-running systems. For any stored data item $m$, the corresponding state of a node $v$ can be recreated on a (new) node $v'$ by retrieving $m$ and then deriving the parts of $m$ that $v'$ must store, discarding the other parts. With this approach the entire state of the storage network must be sent to $v'$, which is practically infeasible.
Thus, a \emph{recovery algorithm} is needed that is efficient in the sense that the amount of transferred data exceeds the state size of $v$ merely by a small factor.

The quality of AVID protocols is measured primarily along two dimensions: how much data the storage nodes store and how much data must be transferred when reading or writing, or when recovering the state of a failed node. In this context, the \emph{space complexity} is formally defined as follows.

\begin{figure*}[t]
\vspace{.4cm}
    \centering\small
    \hspace{1.2em}%
    % --- Matrix A ---
    \begin{subfigure}[b]{0.20\textwidth}
        \centering
        \[A \coloneqq 
\begin{pNiceArray}{ccc}[margin]
  m_{1,1} & \cdots & m_{1,n-t} \\
  \vdots  & \ddots & \vdots  \\
  m_{n-t,1} & \cdots & m_{n-t,n-t} 
\CodeAfter
  \begin{tikzpicture}[
      back brace/.style={decorate, decoration={brace, amplitude=3pt}}]
    
    % Top brace: from top-left of (1,1) to top-right of (1,3)
    \draw [back brace] ([yshift=4pt]1-1.north west) -- ([yshift=4pt]1-3.north east) 
          node [midway, above=3pt] {$n-t$};
    
    % Right brace: from top-right of (1,3) to bottom-right of (3,3)
    \draw [back brace] ([xshift=15pt]1-3.north east) -- ([xshift=10pt]3-3.south east) 
          node [midway, right=3pt] {$n-t$};
          
  \end{tikzpicture}
\end{pNiceArray}
\]
    \end{subfigure}
    \hspace{6em}%
    % --- Matrix A' ---
    \begin{subfigure}[b]{0.20\textwidth}
        \centering
        \[
A' \coloneqq 
\begin{pNiceArray}{ccc}[margin]
  \NotEmpty & & \NotEmpty \\ % Row 1 (Top brace & side brace start)
  & A & \\
  &   & \NotEmpty \\ \Hline % Row 3 (Top side brace end)
  \NotEmpty & & \NotEmpty \\ % Row 4 (Bottom side brace start)
  & R_c & \NotEmpty \\ % Row 5 (Bottom side brace end)
\CodeAfter
  \begin{tikzpicture}[remember picture, overlay,
      back brace/.style={decorate, decoration={brace, amplitude=3pt}}]
    
    % Top brace: from 1-1 to 1-3
    \draw [back brace] ([yshift=4pt]1-1.north west) -- ([yshift=4pt]1-3.north east) 
          node [midway, above=3pt] {$n-t$};
    
    % Right brace Top: from 1-3 to 3-3
    \draw [back brace] ([xshift=10pt]1-3.north east) -- ([xshift=10pt]3-3.south east) 
          node [midway, right=4pt] {$n-t$};
    
    % Right brace Bottom: from 4-3 to 5-3
    \draw [back brace] ([xshift=10pt]4-3.north east) -- ([xshift=10pt]5-3.south east) 
          node [midway, right=4pt] {$t$};
          
  \end{tikzpicture}
\end{pNiceArray}
\]
    \end{subfigure}
    \hspace{2em}%
    % --- Matrix M ---
    \begin{subfigure}[b]{0.38\textwidth}
        \centering
        \[
M \coloneqq 
\begin{pNiceArray}{ccc|c|ccccc}[margin]
  % --- TOP BLOCK ---
  \NotEmpty & & \NotEmpty & \NotEmpty & \NotEmpty & & & & \NotEmpty \\ % Row 1
  & A & & ~ R_r & & & ~ S_{r} ~ & & \\                                  % Row 2
  \NotEmpty & & & & & & & & \NotEmpty \\ \Hline                        % Row 3
  % --- BOTTOM BLOCK ---
  \NotEmpty & & & & & & & & \NotEmpty     \\                           % Row 4
  & R_c & & R_{rc} & & & ~ S_{rc} ~ & & \NotEmpty                        % Row 5
\CodeAfter
  \begin{tikzpicture}[remember picture, overlay,
      back brace/.style={decorate, decoration={brace, amplitude=3pt}}]
    
    % --- TOP BRACES ---
    \draw [back brace] ([yshift=12pt]1-1.north west) -- ([yshift=12pt]1-3.north east) 
          node [midway, above=3pt] {$n-t$};
    
    \draw [back brace] ([yshift=12pt, xshift=-0.5em]1-4.north west) -- ([yshift=12pt, xshift=0.5em]1-4.north east) 
          node [midway, above=3pt] {$t$};
          
    \draw [back brace] ([yshift=12pt]1-5.north west) -- ([yshift=12pt]1-9.north east) 
          node [midway, above=3pt] {$n$};

    % --- RIGHT SIDE BRACES ---
    \draw [back brace] ([xshift=15pt]1-9.north east) -- ([xshift=15pt]3-9.south east) 
          node [midway, right=4pt] {$n-t$};
    
    \draw [back brace] ([xshift=15pt]4-9.north east) -- ([xshift=15pt]5-9.south east) 
          node [midway, right=4pt] {$t$};
          
  \end{tikzpicture}
\end{pNiceArray}
\]
    \end{subfigure}
%        \vspace{-.2cm}
		\caption{
		The data item $m$ is embedded in a square matrix $A$ with $n-t$ elements in each row and column.
		The columns of $A$ are $(n, n-t)$-erasure coded and then concatenated to create $A'$.
		The rows of $A'$ are $(2n, n-t)$-erasure coded and then stacked on each other to create the $n \times 2n$-matrix $M$ of size $2(n/(n-t))^2|m|$.}
		\label{fig:encoding}
		%\vspace{-.4cm}
\end{figure*}

\begin{definition}[Space complexity] For any data item $m$, the \emph{space complexity} is the maximum number of bits stored across all honest nodes to provide access to $m$ once dispersal has terminated.
\end{definition}

The space complexity is a function of the network size $n$, message size $|m|$, and cryptographic parameters specific to the protocol.
This definition does not take into account that the memory consumption may be higher during the execution of the dispersal algorithm because nodes may temporarily buffer received data that is not part of their long-term storage. This omission is justified in that the primary concern is to minimize the long-term storage consumption; moreover, the maximum amount of data buffered at any stage is upper bounded by the amount of data transmitted, which is captured in the following definition.

\begin{definition}[Communication complexity] For any operation, the \emph{communication complexity} is the maximum total number of bits transmitted by all honest nodes (and the initiating client, if any) to carry out said operation.
\end{definition}

As mentioned above, the relevant operations are dispersal, retrieval, and recovery.
For protocols where honest nodes that did not receive their share of data during the dispersal execute the recovery algorithm in order to participate in the retrieval (e.g., \cite{danezis2025}), the question is how to properly define the communication complexity of dispersal and retrieval.
Since these recovery operations occur only once during the first retrieval, we count the communication complexity of the required recovery operations as part of the dispersal. An example of this calculation is given in \S\ref{sec:related_work}.

In addition to the space and communication complexity, we also consider the \emph{round complexity}, i.e., the maximum number of communication rounds required to carry out each operation, where each round consists of the steps of sending messages, receiving messages of this round, and then carrying out some local computation.

The primary technique to keep these complexities low in all recent papers on AVID (and related problems such as reliable broadcast) is \emph{erasure coding}. An $(n, k)$-erasure code is used to encode any message $m$ into $n$ so-called \emph{fragments} $f_1, \ldots, f_n$ in such a way that any subset of $k$ fragments can be used to decode message $m$.
In this paper, the functions to encode data and decode fragments are called \texttt{encode} and \texttt{decode}, respectively.
Concretely, given a vector $m=(m_1,\ldots,m_k)$ of $k$ data elements, where each $m_i$ is a string of bits, and a parameter $n$, \texttt{encode}$(m,n)$ returns a vector $v'=(f_1,\ldots,f_n)$ of $n$ fragments.
The function \texttt{decode} takes any subset of $k$ distinct fragments $f_{i_1},\ldots, f_{i_k}$ as input and outputs a vector $(m_1,\ldots,m_k)$.
Without loss of generality, we can assume that the used erasure code is \emph{systematic}, which means that the output of the \texttt{encode} function contains the input vector. Concretely, we assume in the following that $m_i = f_i$ for all $i \in [k]$, where $[k] \coloneqq \{1,\ldots, k\}$. The other fragments $f_{k+1},\ldots, f_n$ are called \emph{recovery fragments}.
Moreover, we consider \emph{linear} erasure codes, where each encoded symbol is a linear combination of the original message symbols over a finite field. 
Given an input vector $m=(m_1,\ldots,m_k)$ of size $|m|$, the size of a fragment $f_i$ is $\max(|m|/k,\lceil\log(n)\rceil)$. The size is at least $\lceil\log(n)\rceil$ due to the fact that erasure coding works with a field of cardinality at least $n$. In the following, we assume that $|m|/k \ge \lceil\log(n)\rceil$ and thus $|f_i| = |m|/k$ for all $i \in [n]$. This assumption already holds for moderately large $|m|$ and practical values of $n$ but it can further be justified in that fragments are always accompanied by a corresponding commitment proof (discussed below) whose size exceeds the symbol size of erasure coding for known commitment schemes.

An asymmetric signature scheme is assumed to be in place and all exchanged messages are authenticated. The only other cryptographic tool that is required is a vector commitment scheme to prevent Byzantine nodes from disseminating invalid fragments, which would foil any attempt to reconstruct $m$. Concretely, we use vector commitments with the following properties. First, given a vector $d$ of data elements, any node can deterministically create a binding commitment $\mathcal{C}_d$ to $d$. Second, given vector $d$ and an index $i$, any node can generate a proof $\pi_d[i]$ attesting that $d[i]$ is the $i^\mathit{th}$ data element of $d$.
In other words, once a node has obtained the full vector, it can generate all proofs for data elements in vector $d$ itself.
Naturally, proofs can be evaluated against a commitment $\mathcal{C}_d$, verifying that the proof is valid for the committed vector $d$.
The \texttt{gen\_proofs} function used in the protocols performs $O(n\log n)$ hash evaluations for Merkle-tree commitments, and $O(n\log n)$ field operations for KZG commitments; this computational overhead is small for practical values of $n$ and can therefore be considered to have a negligible effect on the latency of all operations in a real-world deployment.
 
These concepts are extended to two dimensions in our constructions: data is erasure-coded into a matrix $M$, such that given sufficiently many (i.e., $n-t$) elements of a row or column of $M$, the entire row or column can be derived using the function \texttt{decode}.
As the entire protocol deals with matrix $M$, we unify the terms ``element'' and ``fragment'' and solely use ``element'' throughout the remainder of the paper.
In order to protect the integrity of $M$, a commitment $\mathcal{C}_M$ to $M$ and a matrix of proofs $\pi_M$ is created where $\pi_M[i,j]$ proves that $M[i,j]$ is the correct element at indices $i$ and $j$ under commitment $\mathcal{C}_M$. Given sufficiently many (i.e., $n-t$) elements of a row or column of $M$ and their proofs, the commitment scheme must support constructing proofs for the full row or column using the function \texttt{gen\_proofs}.
Such a commitment can be implemented, e.g., using a Merkle tree, where a set of specific inner nodes, one per height, constitutes a membership proof. Given that the height is logarithmic in $n$, the size $|\pi|$ of a proof $\pi$ is $\Theta(\kappa\log(n))$ where $\kappa$ is the output size of the used cryptographic hash function. Alternatively, a polynomial commitment scheme such as the KZG scheme~\cite{kate2010} can be used, where commitments and proofs correspond to a single group element. As mentioned above, since such commitment schemes operate over large groups, we assume that the commitment and proof sizes are bounded by $\Theta(\log(n))$.
In contrast to the approach using Merkle trees, this scheme requires a trusted setup and additional cryptographic assumptions. A detailed description of these schemes is considered out of scope, i.e., we simply assume that such a scheme is given.

% !TEX root = ./avid_spaa.tex

\section{Main Protocol}
\label{sec:avid}

\subsection{Overview}

The proposed AVID protocol, denoted by $\mathcal{P}$, consists of distributed algorithms for dispersal, retrieval, and recovery.
At a high level of abstraction, the data item $m$ is encoded into a matrix $M$ such that $m$ can be reconstructed when given sufficiently many row (column) elements for a certain number of rows (columns).
For all $i \in [n]$, each node $v_i$ is supposed to hold \emph{row data} $r_i \coloneqq M[i,1..n-t]$ and \emph{column data} $c_i \coloneqq M[1..n-t,i]$, i.e., the first $n-t$ elements of row $i$ and column $i$, respectively. Additionally, $v_i$ must store the corresponding membership proofs $\pi_i^r \coloneq \pi_M[i,1..n-t]$ and $\pi_i^c \coloneq \pi_M[1..n-t,i]$.
The dispersal algorithm disseminates this \emph{storage data} to all nodes, which constitutes a write operation, whereas the retrieval algorithm is executed to read data by gathering row data from the storage nodes and reconstructing $m$. The challenge is to define algorithms that require a minimal amount of data to be stored and transmitted for the operations defined in \S\ref{sec:model}.

\textbf{Encoding}:
The encoding of $m$ into an $n \times 2n$-matrix $M$, illustrated in Figure~\ref{fig:encoding}, starts by converting $m$ into an $(n-t) \times (n-t)$ matrix. This matrix is first expanded by applying an  $(n, n-t)$-erasure code column-wise and then by applying an $(2n, n-t)$-erasure code row-wise to obtain matrix $M$. % as illustrated in Figure~\ref{fig:encoding}.
Note that no node or client ever explicitly stores the entire matrix, and specific parts of $M$ are only required for certain operations.

\textbf{Dispersal}:
The dispersal algorithm ensures that every honest node $v_i$ stores its row and column data (including membership proofs).
To this end, the client sends the storage data to each node individually and, subsequently, transfers additional information to all nodes, which engage in a multi-round protocol to achieve a full dispersal despite Byzantine storage nodes or a Byzantine client $c$. Every element of $M$ may be transferred during dispersal, in particular elements in $S_{r}$ and $S_{rc}$ are \emph{only} used for this operation.

\begin{figure}[t]
  \centering
    \begin{subfigure}[c]{0.4\columnwidth}
        \centering
        \includegraphics[width=0.5\textwidth]{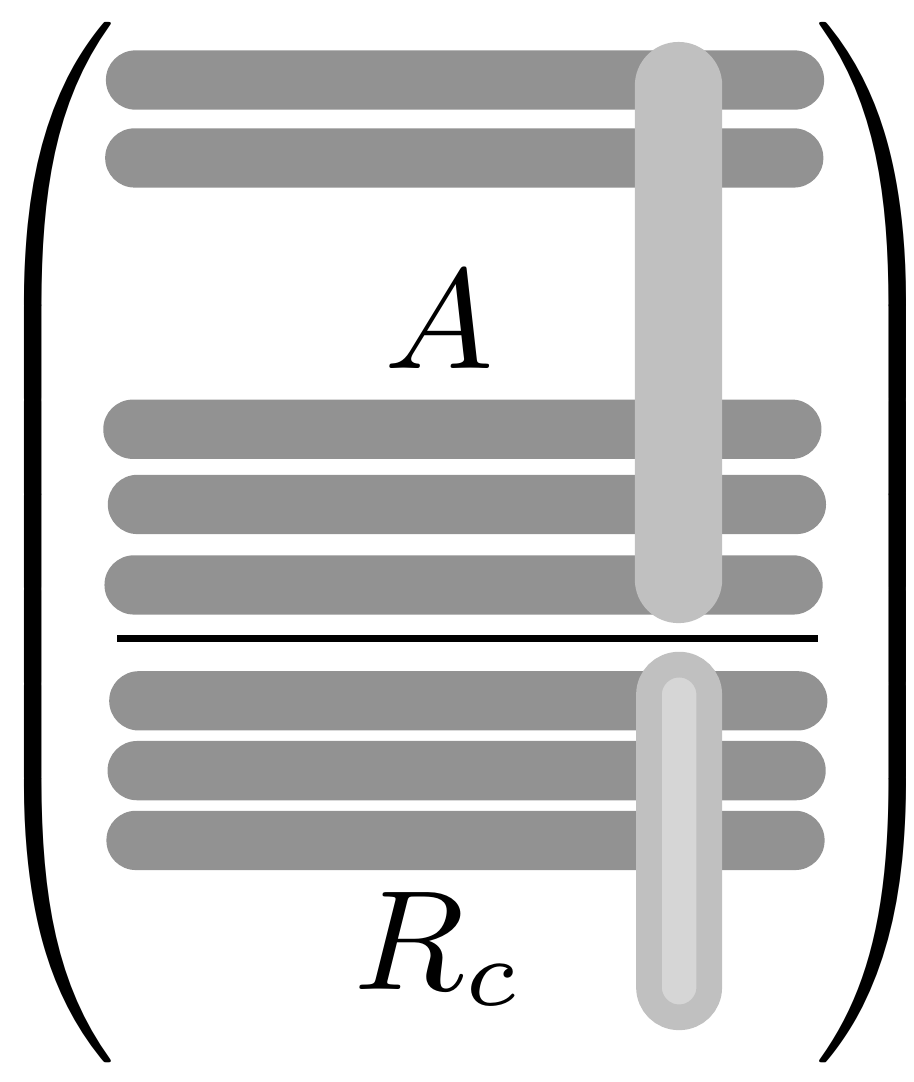}
        \\(a)
\end{subfigure}
    \hfill
    \begin{subfigure}[c]{0.5\columnwidth}
        \centering
       \includegraphics[width=0.80\textwidth]{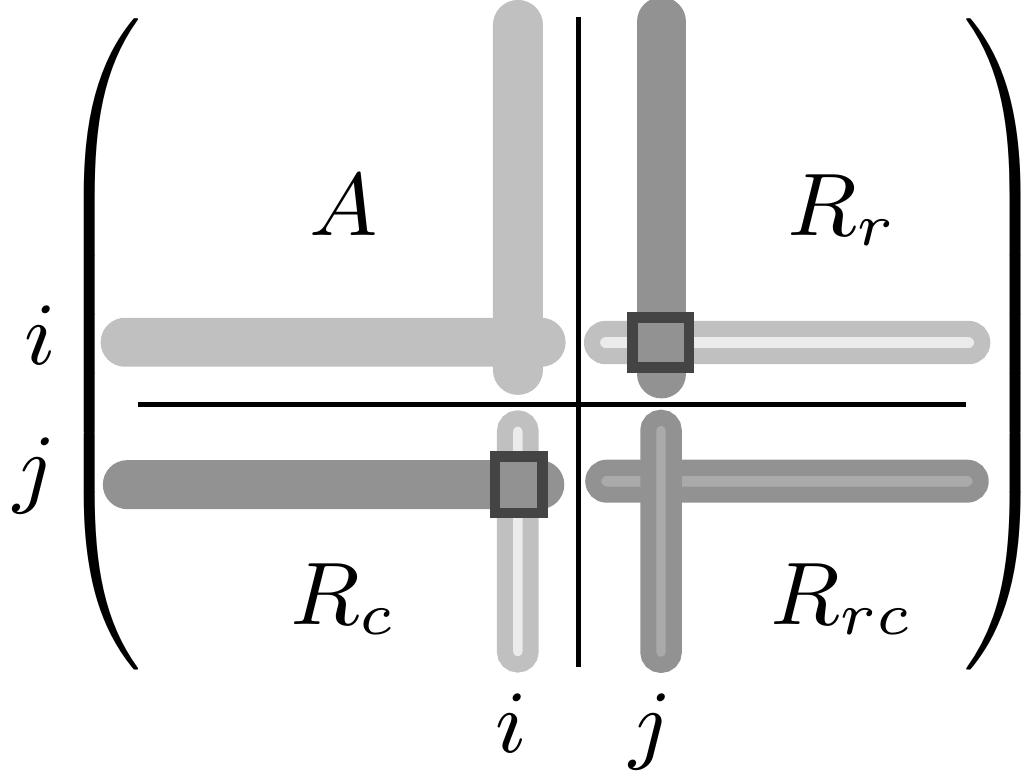}
               \\\vspace{-.1cm}(b)
\end{subfigure}\vspace{-.2cm}
\caption{
(a) Retrieval:
The client receives $n-t$ row elements from $n-t$ honest nodes, thus it has enough elements from all columns of $A'$ to decode them and reconstruct $m$.\newline
(b) Recovery: Node $v_j$ has an element in its column (row) that is in $v_i$'s row (column). Sending the common elements to the replacment node $v_i'$ enables $v_i'$ to decode the first $n-t$ elements of both its row and column.
}
\label{fig:algo}
\end{figure}

\textbf{Retrieval}:
When a client requests to read data item $m$, all honest nodes send their row data. The collection of this data enables the client to obtain $m$ by decoding the columns of $A'$ as illustrated in Figure~\ref{fig:algo}(a), i.e., only elements in $A$ and $R_c$ are used.

\textbf{Recovery}:
If node $v_i'$ is meant to replace (failed) node $v_i$, every honest node $v_j$ sends the $i^\mathit{th}$ element of its row and column data after an encoding step. When enough nodes reply with these pairs of elements, $v_i'$ can reconstruct its row and column data as depicted in Figure~\ref{fig:algo}(b), making use of elements in $A$, $R_r$, $R_c$, and $R_{rc}$.

The exchanged information is augmented with commitments to cope with Byzantine behavior.
Before describing the algorithms executed by the clients and storage nodes in more detail, the data encoding procedure is explained more formally.

\subsection{Data Encoding and Integrity}
The dispersal and retrieval algorithms require some helper functions for the encoding and decoding at the client and to ensure data integrity, which are collected in Algorithm~\ref{algo:helper}.
The function \texttt{create\_matrix} converts the data item $m$ into an $n \times 2n$-matrix as follows. First, $m$ is split into $(n-t)^2$ data elements $m_{i,j}$ for $i,j \in [n-t]$. These elements are then placed into an $(n-t) \times (n-t)$-matrix $A$.
Next, matrix $A$ is expanded in both dimensions by first adding $t$ recovery elements to every column and then extending every row with $n+t$ recovery elements as discussed in the previous section.
The operator $||$ in the pseudocode of \texttt{create\_matrix} concatenates the columns. Note that the input, i.e., the columns, as well as the output of \texttt{encode} are transposed because the input and output of \texttt{encode} are both defined as row data.

\begin{algorithm}[t]
\begin{algorithmic}[1]
\small
\caption{Client helper functions.}
    \label{algo:helper}
\Function{\texttt{create\_matrix}}{$m$}
\State Split $m$ into $(n-t)^2$ elements $m_{i,j}$ for $i,j \in [n-t]$
\State Place elements into $(n-t,n-t)$-matrix $A$, with $a_{i,j}  \coloneqq m_{i,j}$
\State // Add $t$ recovery elements to every column to create $A'$
\State $A' \coloneqq (\texttt{encode}(A[\cdot, 1]^T,n)^T ~||\ldots ||$ \texttt{encode}$(A[\cdot, n-t]^T,n)^T)$
\State // Add $n+t$ recovery elements to every row to create $M$
\State $M \coloneqq (\texttt{encode}(A'[1,\cdot],2n), \ldots,  \texttt{encode}(A'[n,\cdot],2n))$
\State \Return $M$
\EndFunction
\vspace{0.1cm}
%\State
\Function{\texttt{extract}}{$A$}
\State $m \coloneqq \bot$
\For {$i \in [n-t]$}
\For {$j \in [n-t]$}
\State $m \coloneqq m \;||\; A[i,j]$
\EndFor
\EndFor
\State \Return $m$
\EndFunction
\end{algorithmic}
\end{algorithm}

This encoding results in a matrix with $2n^2$ elements satisfying the property that any set of elements that contains at least $n-t$ rows with at least $n-t$ elements each, can be used to recover $m$ with row-wise decoding. 
Analogously, any set of at least $n-t$ columns with at least $n-t$ elements each can be used to recover $m$ with column-wise decoding, respectively.
Row-wise decoding consists of calling the \texttt{decode} function over rows first, followed by a call of the \texttt{decode} function over the outputs.
Analogously, column-wise decoding first decodes the columns and then the outputs.
By using a suitable generator matrix, we ensure that if we encode row data into $n$ elements instead of $2n$, the $i^\mathit{th}$ element is the same in both encodings for any $i \in [n]$.

The function \texttt{extract} reads the entries of $M$ at indices $i,j \in [n-t]$ and concatenates them to reconstruct $m$. Note that $M$ contains the parts of $m$ itself at these indices because the encoding is assumed to be systematic as defined in \S\ref{sec:model}.

\subsection{Dispersal}
\label{sec:dispersal}

The dispersal algorithm is initiated by a client with the aim to write data item $m$. First, the client computes the encoding matrix $M$ and a commitment and proofs for it. Then, the following high-level steps, depicted in Figure~\ref{fig:dispersal}, ensure that every honest node $v_i$ stores its storage data $(r_i, c_i, \pi_i^r, \pi_i^c$) at the end.
\begin{enumerate}
\item The client sends row and column data of $M$ to each node.
\item Nodes validate the data and return \emph{ack} messages.
\item The client initiates a reliable broadcast of the set $S$ of received \emph{ack} messages. Let $\bar{V}$ be the set of nodes from which no \emph{ack} message was received. The client further distributes \emph{spread} messages containing elements of the rows for nodes in $\bar{V}$.  
\item Nodes forward row elements to nodes in $\bar{V}$ in \emph{row\_info} messages and broadcast \emph{info\_sent} messages afterwards.
\item Nodes in $\bar{V}$ decode their row data if necessary. All nodes send
\emph{column\_info} messages containing column elements to the nodes in $\bar{V}$ so they can decode their column data. Lastly, \emph{ready} messages are used to trigger the termination of dispersal.
\end{enumerate}

\begin{figure}[t]
\centering
\footnotesize
% !TEX root =  ../avid_spaa.tex
\begin{tikzpicture}[scale=0.78]
	\begin{pgfonlayer}{nodelayer}
		\node [style=none, label={left:$c$}] (0) at (-33, 6) {};
		\node [style=none, label={left:$v_1$}] (1) at (-33, 5) {};
		\node [style=none, label={left:$v_2$}] (2) at (-33, 4) {};
		\node [style=none] (3) at (-33.5, 3) {...};
		\node [style=none, label={left:$v_{n-1}$}] (4) at (-33, 2) {};
		\node [style=none, label={left:$v_{n}$}] (5) at (-33, 1) {};
		\node [style=none] (6) at (-13, 6) {};
		%\node [style=none] (7) at (-9, 5) {};
		%\node [style=none] (8) at (-9, 4) {};
		%\node [style=none] (9) at (-9, 3) {};
		%\node [style=none] (10) at (-9, 2) {};
		%\node [style=none] (11) at (-9, 1) {};
		\node [style=none] (23) at (-31, 6.75) {$\mathit{disperse}$};
		\node [style=none] (24) at (-27, 6.75) {$\mathit{ack}$};
		\node [style=none] (25) at (-23, 7.25) {$\mathit{RB}$};
		\node [style=none] (30) at (-23, 6.5) {$\mathit{spread}$};
		\node [style=none] (31) at (-19, 7.25) {$\mathit{row\_info}$};
		\node [style=none] (32) at (-19, 6.5) {$\mathit{info\_sent}$};
		\node [style=none] (34) at (-15, 7.25) {$\mathit{col\_info}$};
		\node [style=none] (35) at (-15, 6.5) {$\mathit{ready}$};
%		\node [style=none] (36) at (-11, 6.75) {$\mathit{receipt}$};
		\node [style=none] (37) at (-29, 5) {};
		\node [style=none] (38) at (-29, 4) {};
		\node [style=none] (41) at (-25, 6) {};
		\node [style=none] (42) at (-23.75, 5) {};
		\node [style=none] (43) at (-23.75, 4) {};
		\node [style=none] (44) at (-23.75, 2) {};
		\node [style=none] (45) at (-23.75, 1) {};
		\node [style=none] (46) at (-21, 5) {};
		\node [style=none] (47) at (-21, 4) {};
		\node [style=none] (51) at (-17, 5) {};
		\node [style=none] (53) at (-17, 4) {};
		\node [style=none] (54) at (-18.5, 2) {};
		\node [style=none] (55) at (-17, 2) {};
		\node [style=none] (56) at (-18.5, 1) {};
		\node [style=none] (57) at (-17, 1) {};
		\node [style=none] (58) at (-14.5, 2) {};
		\node [style=none] (59) at (-15, 1) {};
		\node [style=none] (60) at (-13, 5) {};
		\node [style=none] (61) at (-13, 4) {};
		\node [style=none] (62) at (-13, 2) {};
		\node [style=none] (63) at (-13, 1) {};
		\node [style=none] (64) at (-17, 5) {};
		\node [style=none] (65) at (-17, 4) {};
		\node [style=none] (66) at (-17, 2) {};
		\node [style=none] (67) at (-17, 1) {};
		\node [style=none] (68) at (-16.25, 5) {};
		\node [style=none] (69) at (-15.75, 4) {};
		\node [style=none] (70) at (-16, 2) {};
		\node [style=none] (71) at (-16.25, 1) {};
		\node [style=none] (73) at (-29, 1) {};
		\node [style=none] (74) at (-29, 3.25) {};
		\node [style=none] (75) at (-28.75, 1.25) {};
		\node [style=none] (76) at (-29.25, 1.25) {};
		\node [style=none] (77) at (-29.25, 0.75) {};
		\node [style=none] (78) at (-28.75, 0.75) {};
		\node [style=none] (79) at (-29, 2) {};
		\node [style=none] (80) at (-28.75, 2.25) {};
		\node [style=none] (81) at (-29.25, 2.25) {};
		\node [style=none] (82) at (-29.25, 1.75) {};
		\node [style=none] (83) at (-28.75, 1.75) {};
		\node [style=none] (84) at (-21, 1) {};
		\node [style=none] (85) at (-20.75, 1.25) {};
		\node [style=none] (86) at (-21.25, 1.25) {};
		\node [style=none] (87) at (-21.25, 0.75) {};
		\node [style=none] (88) at (-20.75, 0.75) {};
		\node [style=none] (89) at (-21, 2) {};
		\node [style=none] (90) at (-20.75, 2.25) {};
		\node [style=none] (91) at (-21.25, 2.25) {};
		\node [style=none] (92) at (-21.25, 1.75) {};
		\node [style=none] (93) at (-20.75, 1.75) {};
		\node [style=none] (94) at (-29, 0.75) {};
		\node [style=none] (95) at (-29, 7) {};
		\node [style=none] (96) at (-25, 7) {};
		\node [style=none] (97) at (-25, 0.75) {};
		\node [style=none] (98) at (-21, 7) {};
		\node [style=none] (99) at (-21, 0.75) {};
		\node [style=none] (101) at (-17, 7) {};
		\node [style=none] (102) at (-17, 0.75) {};
		\node [style=none] (103) at (-17, 1) {};
		\node [style=none] (104) at (-13, 7) {};
		\node [style=none] (105) at (-13, 0.75) {};
		\node [style=none] (106) at (-22.5, 2) {};
		\node [style=none] (107) at (-22.5, 1) {};
		\node [style=none] (108) at (-22.5, 4) {};
		\node [style=none] (109) at (-22.5, 5) {};
	\end{pgfonlayer}
	\begin{pgfonlayer}{edgelayer}
		\draw [style=dotted] (0.center) to (6.center);%w line
		\draw [style=dotted] (1.center) to (60.center);%v1 line
		\draw [style=dotted] (2.center) to (61.center);% v2 line
		\draw [style=dotted] (4.center) to (62.center);%vn-1 line
		\draw [style=dotted] (5.center) to (63.center);
		\draw [style=data] (0.center) to (37.center);%disperse
		\draw [style=data] (0.center) to (38.center);%disperse
		\draw [style=control] (37.center) to (41.center);
		\draw [style=control] (38.center) to (41.center);
		\draw [style=data] (41.center) to (46.center);%spread
		\draw [style=data] (41.center) to (47.center);%spread
		\draw [style=control] (41.center) to (42.center);
		\draw [style=control] (41.center) to (43.center);
		\draw [style=control] (41.center) to (44.center);
		\draw [style=control] (41.center) to (45.center);
		\draw [style=data] (46.center) to (54.center);%row_info
		\draw [style=data] (47.center) to (56.center);%row_info
		\draw [style=control] (46.center) to (55.center);%info sent
		\draw [style=control] (47.center) to (57.center);%info sent
		\draw [style=control] (47.center) to (55.center);%info sent
		\draw [style=control] (47.center) to (51.center);%info sent
		\draw [style=control] (46.center) to (57.center);%info sent
		\draw [style=control] (46.center) to (53.center);%info sent
		\draw [style=data] (64.center) to (58.center);%col_info
		\draw [style=data] (64.center) to (59.center);%col_info
		\draw [style=data] (65.center) to (58.center);%col_info
		\draw [style=data] (65.center) to (59.center);%col_info
		\draw [style=data] (66.center) to (59.center);%col_info
		\draw [style=data] (67.center) to (58.center);%col_info
		\draw [style=control] (51.center) to (61.center);%ready
		\draw [style=control] (51.center) to (62.center);%ready
		\draw [style=control] (51.center) to (63.center);%ready
		\draw [style=control] (53.center) to (60.center);%ready
		\draw [style=control] (53.center) to (62.center);%ready
		\draw [style=control] (53.center) to (63.center);%ready
		\draw [style=control] (55.center) to (60.center);%ready
		\draw [style=control] (55.center) to (61.center);%ready
		\draw [style=control] (55.center) to (63.center);%ready
		\draw [style=control] (57.center) to (60.center);%ready
		\draw [style=control] (57.center) to (61.center);%ready
		\draw [style=control] (57.center) to (62.center);%ready
%		\draw [style=control, color=green] (60.center) to (6.center);%receipt
%		\draw [style=control, color=yellow] (61.center) to (6.center);%receipt
%		\draw [style=control, color=blue] (62.center) to (6.center);%receipt
%		\draw [style=control, color=red] (63.center) to (6.center);%receipt
		\draw [style=data] (47.center) to (54.center);%row info
		\draw [style=data] (46.center) to (56.center);% row info
		\draw [style=control] (70.center) to (61.center);
		\draw [style=control] (69.center) to (63.center);
		\draw [style=control] (74.center) to (41.center);
		\draw [style=data] (0.center) to (74.center);
		\draw [style=data fail] (0.center) to (73.center);
		\draw [style=cross] (76.center) to (78.center);
		\draw [style=cross] (77.center) to (75.center);
		\draw [style=cross] (81.center) to (83.center);
		\draw [style=cross] (82.center) to (80.center);
		\draw [style=data fail] (0.center) to (79.center);
		\draw [style=cross] (86.center) to (88.center);
		\draw [style=cross] (87.center) to (85.center);
		\draw [style=cross] (91.center) to (93.center);
		\draw [style=cross] (92.center) to (90.center);
		\draw [style=data fail] (41.center) to (89.center);
		\draw [style=data fail] (41.center) to (84.center);
		\draw [style=dotted] (95.center) to (94.center);%line disperse | ack
		\draw [style=dotted] (96.center) to (97.center);%line ack | RB
		\draw [style=dotted] (98.center) to (99.center);%line RB | row info
		\draw [style=dotted] (101.center) to (102.center);% line row info | ready
%		\draw [style=dotted] (104.center) to (105.center);
		\draw [style=control] (45.center) to (106.center);%RB
		\draw [style=control] (44.center) to (107.center);%RB
		\draw [style=control] (45.center) to (108.center);%RB
		\draw [style=control] (43.center) to (107.center);%RB
		\draw [style=control] (43.center) to (109.center);%RB
		\draw [style=control](42.center) to (106.center);%RB
		\draw [style=control](44.center) to (108.center);%RB
		\draw [style=control] (44.center) to (109.center);%RB
		\draw [style=control] (42.center) to (108.center);%RB
	\end{pgfonlayer}
\end{tikzpicture}
\vspace{-.8cm}
\caption{
	Dispersal algorithm to disseminate row and column data.
	Black, solid arrows and gray, dashed arrows represent data and control messages, respectively.
	When nodes obtain enough \emph{ready} messages, they are guaranteed to have their storage data or receive it eventually.
	The reliable broadcast (RB) step requires multiple communication rounds, and all steps may be interleaved due to asynchrony.
}\label{fig:dispersal}
\end{figure}
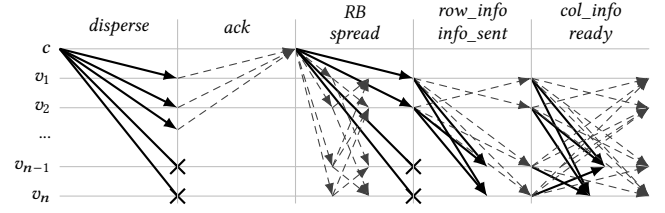

After step (2), the nodes in $V\setminus\bar{V}$, i.e., the nodes from which the client has received an \emph{ack} message, possess or can derive one element each of the row and column of every other node, in particular of the nodes in $\bar{V}$.
Concretely, node $v_k \in V\setminus\bar{V}$ can extract $M[j,k]$ (the $k^{th}$ element of $v_j$'s row) from its own column data and $M[k,j]$ (the $k^{th}$ element of $v_j$'s column) from its own row data.
This is the key insight behind the improved communication complexity: because each node in $V\setminus\bar{V}$ already carries one element of every other node's row and one element of every other node's column, the client does \emph{not} need to retransmit that data. It only needs to supply the \emph{second half} of the row (the recovery elements beyond position $n$) for each of the $t$ nodes in $\bar{V}$.
The nodes in $V\setminus\bar{V}$ can use this data to help the nodes in $\bar{V}$ obtain their row and column data.
However, only $n-2t$ of the nodes in $V\setminus\bar{V}$ may be honest, so this data is not sufficient to ensure that all honest nodes can reconstruct their storage data.
Therefore, the client sends additional row elements to the nodes, which they then forward to the nodes in $\bar{V}$.
A Bracha-style broadcast is used with \emph{info\_sent} and \emph{ready} messages in steps (3-5) to successfully conclude the dispersal. Note that all steps require one communication round except for the reliable broadcast of set $S$, which needs two or more rounds.
The concrete steps at the client and the storage nodes are discussed next.

\subsubsection{Dispersal at Clients}

The actions of client $c$ for dispersals are shown in Algorithm~\ref{algo:client_disperse}.
When dispersing data item $m$, the first step is to construct the $n \times 2n$-matrix $M$ using the function \texttt{create\_matrix}.
Next, the client creates a commitment $\mathcal{C}_M$ for the matrix and proofs $\pi_M $ for all rows and columns.
The identifier for $m$ is defined as $id \coloneqq \mathcal{C}_M$, i.e., the commitment to $M$. This identifier is added as a parameter to all messages so that each message can be uniquely associated with a specific execution of the algorithm. Furthermore, we assume that every node always signs its messages of any type.

\begin{algorithm}[t]
\begin{algorithmic}[1]
\small
\caption{Dispersal at client $c$ of data item $m$ with identifier $id$.}
    \label{algo:client_disperse}
%\Function{\texttt{disperse}}{$m$}
\State $M \coloneqq$ \texttt{create\_matrix}($m$)
\State $(\mathcal{C}_M, \pi_M) \coloneqq$ \texttt{commitment\_proofs}($M$)
\State $id \coloneqq$ $\mathcal{C}_M$
\For {$i \in [n]$}
\State // Information to be stored by node $v_i$
\State $(r_i, c_i) \coloneqq (M[i,1..n-t], M[1..n-t,i])$
\State $(\pi^r_{i}, \pi^c_{i}) \coloneqq (\pi_M[i,1..n-t], \pi_M[1..n-t,i])$
\State \textbf{send} \emph{disperse}$(id, r_i, c_i,\pi^r_{i},\pi^c_{i})$ to $v_i$
\EndFor
\Upon {received set $S$ of $n-t$ valid \emph{ack} messages for $id$}
\State \textbf{reliable-broadcast} \emph{stored}$(id, S)$
\For {$i \in [n]$}
\State $payload_i \coloneqq \{ (j, M[j,i+n], \pi_M[j,i+n]) |  j: (v_j,\_) \notin S\}$
\State \textbf{send} \emph{spread}($id, payload_i$) to $v_i$
\EndFor
\EndUpon
%\EndFunction
\end{algorithmic}
\end{algorithm}

The client sends $(r_i, c_i, \pi_i^r, \pi_i^c)$ directly to each node $v_i$ in a \emph{disperse} message and waits for the corresponding \emph{ack} messages confirming the receipt of the tuple with a signature. As there may be up to $t$ Byzantine failures, client $c$ merely waits for $n-t$ valid \emph{ack} messages. Each such message contains a pair $(v_i, \sigma_i)$, where $\sigma_i$ denotes $v_i$'s signature over its own identifier and $id$, confirming that $v_i$ indeed sent this \emph{ack} message.
In the following, $S$ denotes the set containing the received $(v_i, \sigma_i)$ pairs from $n-t$ distinct nodes.

Since client $c$ must ensure that all honest nodes receive their data, additional steps
are needed, which are executed as soon as $n-t$ signed \emph{ack} messages have been received. Note that the \textbf{upon} clause in Algorithm~\ref{algo:client_disperse} and the \textbf{upon} clauses in all subsequent algorithms are always executed once for every execution of the algorithm as soon as the required conditions are met.
First, client $c$ uses a reliable-broadcast protocol to disseminate the set $S$ in a \emph{stored} message, which ensures that all honest nodes eventually learn which $n-t$ nodes claim to have received their data.
Next, client $c$ generates the payload of a dedicated \emph{spread} message per node. Specifically, the \emph{spread} message for node $v_i$ contains $M[j, i+n]$, and the corresponding proofs  $\pi_M[j, i+n]$ for all $j$ where $(v_j, \_) \notin S$. The dissemination of these messages concludes the dispersal algorithm for the client.

\subsubsection{Dispersal at Storage Nodes}

The storage nodes participate in the dispersal operation as summarized in Algorithm~\ref{algo:disperse}.
When a node $v_k$ receives a \emph{disperse} message containing
a tuple $(r_k,c_k,\pi_k^r,\pi_k^c)$, it verifies its validity, i.e., the row and column data are consistent with the proofs under commitment $\mathcal{C}_M$. If this is the case, the tuple is stored and an \emph{ack}($id$) message is sent to the client.

Recall that an honest client reliably broadcasts a set $S$ containing node identifiers and signatures after collecting $n-t$ valid \emph{ack} messages. When node $v_k$ receives $S$, $v_k$ determines its validity by verifying all contained signatures and then stores it locally. Note that if an honest node stores $S$, all other honest nodes will eventually store exactly the same set due to the properties of reliable broadcast.
When node $v_k$ receives a \emph{spread} message, the message is first validated: It is considered valid if it contains $t$ tuples $(j, M[j,k+n], \pi_M[j,k+n])$ for distinct indices $j$, where each element $M[j,k+n]$ is consistent with the proof $\pi_M[j,k+n])$. 
Moreover, for each index $j$ it holds that $(v_j,\_) \notin S$. The \emph{spread} message is only evaluated when $S \ne \bot$ and $\texttt{get\_column\_data}(id) \neq \bot$, i.e., the reliable-broadcast message and the column data have been received.
Next, node $v_k$ encodes its stored column data into a vector $c_k'$ of length $n$ and then sends $c_k'[j]$ (which corresponds to $M[j,k]$), $M[j,k+n]$ from the spread message, and the proofs for both elements to $v_j$ in a \emph{row\_info} message. Node $v_k$ then broadcasts \emph{info\_sent}$(id)$ to inform all nodes about the fact that it sent row information to the $t$ nodes that may not have received their data yet.

\begin{algorithm}[t]
\begin{algorithmic}[1]
\small
\caption{Dispersal at node $v_k$ for $id=$\texttt{ID($m$)}. Initially, $S = \bot$.}
    \label{algo:disperse}
\Upon{received valid \emph{disperse}$(id, r_k, c_k, \pi^r_{k}, \pi^c_{k})$}
\State \texttt{set\_storage\_data}$(id, r_k, c_k, \pi^r_{k}, \pi^c_{k})$
\State \textbf{send} \emph{ack}$(id)$ to client
\EndUpon
%\State
\Upon{received valid \emph{stored}$(id, S')$}
\State $S \coloneqq S'$
\EndUpon
%\State
\Upon{received valid \emph{spread}$(id, payload_k)$ \textbf{and} $S \ne \bot$  \textbf{and} \Statex \hspace{4.5pt} $(c_k, \pi^c_{k}) \coloneq\texttt{get\_column\_data}(id) \neq \bot$}
\State $c_k' \coloneqq \texttt{encode}(c_k, n)$
\State ${\pi^c_k}' \coloneqq $\texttt{gen\_proofs}$(c_k, {\pi^c_k})$
\For {$(j, M[j,k+n], \pi_M[j,k+n]) \in payload_k$ and $(v_j, \_) \notin S$}
\State $row\_payload_j = ({c_k}'[j], M[j,k+n], {\pi^c_k}'[j], \pi_M[j,k+n])$
\State \textbf{send} \emph{row\_info}$(id, row\_payload_j)$ to $v_j$
\EndFor
\State \textbf{broadcast} \emph{info\_sent}$(id)$
\EndUpon
%\State
\Upon{received $n-2t$ valid \emph{row\_info} messages \textbf{and} \Statex \hspace{4.5pt} \texttt{get\_row\_data}$(id) = \bot$}
\State $r_k \coloneqq \texttt{decode}(F)$, for set $F$ of row elements in \emph{row\_info} messages
\State ${\pi^r_k}' \coloneqq \texttt{gen\_proofs}(r_k,\Pi)$, for set $\Pi$ of proofs in \emph{row\_info} messages
\State \texttt{set\_row\_data}$(id, r_k, {\pi^r_k}'[1..n-t])$
\EndUpon
%\State
\Upon{$(r_k, \pi^r_{k}) \coloneqq$ \texttt{get\_row\_data}$(id) \ne \bot$ \textbf{and} $S \ne \bot$}
\State $r_k' \coloneqq \texttt{encode}(r_k, n)$
\State ${\pi^r_k}' \coloneqq \texttt{gen\_proofs}(r_k,\pi^r_k)$
\For {$j \in [n]$ where $(v_j, \_) \notin S$}
\State \textbf{send} \emph{column\_info}$(id, r_k'[j],{\pi_k^r}'[j])$ to $v_j$
\EndFor
\EndUpon
%\State
\Upon{received set $C, \Pi$ of matrix elements and their proofs from $n-t$ \Statex \hspace{4.5pt} valid \emph{column\_info} messages  \textbf{and} \texttt{get\_column\_data}$(id) = \bot$}
\State $c_k \coloneqq \texttt{decode}(C)$, for set $C$
\State ${\pi^c_k}' \coloneqq$ \texttt{gen\_proofs}$(c_k, \Pi)$, for set $\Pi$
\State \texttt{set\_column\_data}$(id, c_k, {\pi^c_k}'[1..n-t])$
\EndUpon
%\State
\Upon {received $n-t$ \emph{info\_sent} \textbf{or} $t+1$ \emph{ready} messages}
\State \textbf{broadcast} \emph{ready}$(id)$
\EndUpon
%\State
\Upon {received $n-t$ \emph{ready} messages \textbf{and} \texttt{get\_storage\_data}($id$) $\ne \bot$ \Statex \hspace{4.5pt} \textbf{and} \emph{column\_info} sent}
\State terminate dispersal
\EndUpon
\end{algorithmic}
\end{algorithm}

If $v_k$ does not have its row data $(r_k, \pi_k^r)$, i.e., \texttt{get\_row\_data($id$) $= \bot$}, $v_k$
can construct it after receiving $n-2t$ valid \emph{row\_info} messages, where a \emph{row\_info} message is valid when the contained elements and proofs are consistent. The \texttt{decode} function is used to obtain the vector $r_k$ and the $\texttt{gen\_proofs}$ function is used to construct the vector ${\pi_k^r}' \coloneqq \pi_M[k,1..n]$ of $n$ proofs. The derived pair $(r_k, {\pi_k^r}'[1..n-t])$ is then stored using the function \texttt{set\_row\_data}.

Once the row data and the set $S$ are available, $v_k$ can encode $r_k$ using an $(n, n-t)$-erasure code to obtain the extended vector $r_k'$ of $n$ elements, corresponding to $M[k, 1..n]$. Node $v_k$ then sends $r_k'[j]$ together with the corresponding proof $\pi_k^r[j]$ to $v_j$ in a \emph{column\_info} message for each $v_j$ that does not occur in the set $S$.
As the name suggests, the purpose of this message is to provide elements that can be used to recover a column of matrix $M$: Once $n-t$ valid \emph{column\_info} messages have been received, node $v_k$ can derive column $c_k$ by applying the \texttt{decode} function to the set $C$ of received column elements. The proofs ${\pi_k^c}'$ for all $n$ column elements are then obtained by applying the function \texttt{gen\_proofs}. The column data $c_k$ and the corresponding proofs ${\pi_k^c}'[1..n-t]$ are then stored using the function \texttt{set\_column\_data}.

Node $v_k$ waits for the confirmation that all honest nodes will eventually receive their storage data. To this end, $v_k$ collects \emph{info\_sent} messages. Once it receives at least $n-t$ such messages and validates their signatures successfully, it broadcasts a \emph{ready} message. Alternatively, node $v_k$ also broadcasts a \emph{ready} message if it receives at least $t+1$ \emph{ready} messages itself.
When $v_k$ has received $n-t$ valid \emph{ready} messages, its storage data $(r_k, c_k, \pi_k^r, \pi_k^c)$ is available, and it has sent its \emph{column\_info} messages, then $v_k$ terminates the dispersal.

\subsection{Retrieval}

\subsubsection{Retrieval at Clients}

When retrieving a data item for identifier $id$, the client sends a \emph{retrieve} message containing $id$ to all nodes and waits for $n-t$ valid \emph{response} messages for $id$ from distinct nodes. 
A response is valid if it is signed correctly, it contains a row of $n-t$ elements, and also valid proofs for each of these elements.

As depicted in Figure~\ref{fig:algo}(a), the client can decode column $i$ of matrix $A$ for every $i\in[n-t]$, containing all parts of data item $m$.
Applying the $(n, n-t)$-erasure-encoding function to matrix $A$ first column-wise and then row-wise as in the \texttt{create\_matrix} function yields matrix $M$, which is used to verify the integrity of all data by computing the commitment $\mathcal{C}_M$ and verifying that $\mathcal{C}_M = id$. If this is the case, the data item $m$ is extracted from $A$ and returned. Otherwise, $\bot$ is delivered to indicate that invalid data is stored under the identifier $id$.

\begin{algorithm}[t]
\begin{algorithmic}[1]
\small
\caption{Retrieval at client $c$ for identifier $id$.}
    \label{algo:client_retrieve}
\For {$i \in [n]$}
\State \textbf{send} \emph{retrieve}($id$) to $v_i$
\EndFor
\State \textbf{wait} until received $n-t$ valid responses $\neq \bot$ for $id$
\State Create $(n-t) \times (n-t)$ matrix $R$ out of the received rows
\State $A \coloneqq (\texttt{decode}(R[\cdot, 1]^T)^T \;||\; \ldots \;||\; \texttt{decode}(R[\cdot, n-t]^T)^T)$
\State $A' \coloneqq (\texttt{encode}(A[\cdot, 1]^T, n)^T ~|| \ldots || \texttt{encode}(A[\cdot, n-t]^T, n)^T)$
\State $M \coloneqq (\texttt{encode}(A'[1,\cdot],2n), \ldots, \texttt{encode}(A'[n-t,\cdot],2n))$
\State $(\mathcal{C}_M, \pi_M) \coloneqq$ \texttt{generate}($M$)
\If {$\mathcal{C}_M = id$}
\State $m \coloneqq$ \texttt{extract}($A$)
\State \texttt{deliver}($m$)
\Else
\State \texttt{deliver}($\bot$)
\EndIf
\end{algorithmic}
\end{algorithm}

Constructing the entire matrix $M$ to recompute the commitment for each retrieval may be a burden on lightweight clients. It is possible to circumvent this computation by making the dispersal operation \emph{robust}~\cite{alhaddad2024}, which means that the storage nodes themselves verify that the reconstruction of the data item will succeed. In this case, when retrieving row data from $n-t$ distinct nodes, it suffices to generate $A$ and extract $m$ as the client must have received row data from $n-2t > t$ honest nodes, which implicitly vouch for the validity of the encoding by returning their row data.
This verification can be implemented using fast Reed-Solomon interactive oracle proofs of proximity (FRI)~\cite{ben2018} or inner-product arguments (IPA)~\cite{bootle2016,bunz2018} to prove that each row and column of $M$ corresponds to a polynomial of degree $n-t-1$ (or less), ensuring that $m$ can be uniquely decoded. 
However, these proofs are much more involved and it has yet to be demonstrated that such an approach can reduce the computational overhead in practice. For this reason and also because the addition of robustness would distract from the contributions of this work, we merely state the possibility of introducing such an optimization and omit further details.

\subsubsection{Retrieval at Storage Nodes}

When $v_k$ receives a \emph{retrieve} message from a client, $v_k$
checks whether storage data is available for the requested $id$.
If this is the case, the row data and the corresponding proof data is returned to the client, otherwise $\bot$ is returned.

A client may request the data item before it has been (fully) dispersed. If fewer than $n-t$ honest nodes have obtained their row data, it is possible that the client will not receive sufficiently many rows to reconstruct the corresponding data item $m$.
This is an acceptable outcome as the availability condition of AVID (see Definition~\ref{def:avid}) states that a data item must be received if at least $\ell$, in our case $\ell = n-t$, honest nodes have terminated the dispersal, but not before.
In other words, it is the client's responsibility to drop the retrieval request or resubmit it after a timeout.

\begin{algorithm}[t]
\begin{algorithmic}[1]
\small
\caption{Retrieval at node $v_k$ from client $c$ for identifier $id$.}
    \label{algo:retrieve}
\If {$(r_k, \pi^r_k) \coloneqq$ \texttt{get\_row\_data}$(id) ~\neq \bot$}
\State \textbf{send} \emph{response}($id, r_k, \pi^r_k$) to $c$
\Else
\State \textbf{send} \emph{response}($id,\bot$) to $c$
\EndIf
\end{algorithmic}
\end{algorithm}

\subsection{Recovery}

If a node $v_i$ crashes or leaves the storage network for other reasons, a replacement node $v_i'$ must be instantiated with node $v_i$'s state. In other words, for each data item $m$, it must obtain the $i^\mathit{th}$ row and column data derived from the corresponding matrix $M$. Note that, due to Lemma~\ref{lemma:avid_agreement}, once any honest node terminates the dispersal, all honest nodes will eventually do so; recovery may therefore simply be delayed until the required $n-t$ nodes have their storage data.
%Mechanisms to detect failures and setting up replacement nodes are discussed in Section~\ref{sec:considerations}.
We assume that a mechanism to detect failures and setting up replacement nodes is in place, which triggers the execution of the \texttt{recovery} algorithm in Algorithm~\ref{algo:recover}.

Since the steps are identical for each data item $m$, we focus on a particular $m$ with identifier $id$ and the corresponding storage data $(r_i, c_i, \pi^r_{i}, \pi^c_{i})$ that $v_i'$ is supposed to store.
If node $v_k$ is invoked to help the replacement node $v_i'$ obtain the required state, it loads its storage data $(r_k, c_k, \pi^r_{k}, \pi^c_{k})$. As stated before, $v_k$ will eventually get this data if there is another honest node that receives it as well. Next, $v_k$ encodes $r_k$ into a row $r'_k$ of $n$ elements. This procedure is repeated for $c_k$ to obtain $c'_k$.
Node $v_k$ then sends a \emph{recover} message with the $i^\mathit{th}$ element of the extended row and column to $v_i'$, together with the corresponding commitment proofs. Note the order of elements as the row (column) element of $v_k$ corresponds to a column (row) element of $v_i$.

Node $v_i'$ eventually receives such a message from at least $n-t$ honest nodes and, after validating the message signatures and commitment proofs, it has $n-t$ elements of the $i^\mathit{th}$ row and column. This is sufficient to decode them, resulting in the first $n-t$ elements of the $i^\mathit{th}$ row and column, i.e., $r_i$ and $c_i$.
As far as the commitment proofs are concerned, $v_i'$ only receives $n-t$ proofs instead of $n$. However, having both $r_i$ and $c_i$, plus $n-t$ proofs each for elements in its row and column, the function \texttt{gen\_proofs} can be used again to generate the missing commitment proofs. Thus, $v_i'$ obtains $(r_i, c_i, \pi^r_{i}, \pi^c_{i})$ as desired.

It is worth noting that if there are $t' \le t$ failures at the same time, $v_i'$ will still receive at least $n-t$ \emph{recover} messages due to the fact that we assume that at least $n-t$ nodes neither crash nor deviate from the protocol as defined in \S\ref{sec:model}.

\subsection{Analysis}
\label{sec:analysis}

\subsubsection{Correctness}

First, the correctness of protocol $\mathcal{P}$ is proven in the sense that it has all AVID properties specified in Definition~
\ref{def:avid}.

\begin{algorithm}[t]
\begin{algorithmic}[1]
\small
\caption{Recovery at node $v_k$ for $id$ from $v_i$ to $v_i'$.}
    \label{algo:recover}
\If {$v_i' \neq v_k$}
\State \textbf{wait} until $(r_k, c_k, \pi^r_k, \pi^c_k) \coloneqq$ \texttt{get\_storage\_data}$(id) ~\neq \bot$
\State $r_k' \coloneqq \texttt{encode}(r_k, n)$, $c_k' \coloneqq \texttt{encode}(c_k, n)$
\State ${\pi^r_k}' \coloneqq \texttt{gen\_proofs}(r_k,{\pi^r_k})$, ${\pi^c_k}' \coloneqq \texttt{gen\_proofs}(c_k, {\pi^c_k})$
\State // Extract the elements for $i$ with rows/columns swapped
\State \textbf{send} \emph{recover}($id, c_k'[i], r_k'[i], {\pi^c_k}'[i], {\pi^r_k}'[i]$) to $v_i'$
\Else
\State \textbf{upon} received set $E$ of $n-t$ valid \emph{recover} messages \textbf{do}
\State $r_i \coloneqq $ \texttt{decode}($R$), for set $R$ of $1^{st}$ items of messages in $E$
\State $c_i \coloneqq $ \texttt{decode}($C$), for set $C$ of $2^{nd}$ items of messages in $E$
\State ${\pi_i^r}' \coloneqq $ \texttt{gen\_proofs}($r, \Pi_r$), for set $\Pi_r$ of $3^{rd}$ items of messages in $E$
\State ${\pi_i^c}' \coloneqq $ \texttt{gen\_proofs}($c, \Pi_c$), for  set $\Pi_c$ of $4^{th}$ items of messages in $E$
\State \texttt{set\_storage\_data}$(id, r_i, c_i, {\pi_i^r}'[1..n-t], {\pi_i^c}'[1..n-t])$
\EndIf
\end{algorithmic}
\end{algorithm}
\begin{lemma}[Termination]\label{lemma:avid_termination} If an honest client disperses $m$, then all honest nodes successfully terminate the dispersal of $m$.
\end{lemma}
\begin{proof}
Assume that honest client $c$ disperses data item $m$. Since $c$ is honest, all commitments and proofs that it provides are correct and therefore all messages that $c$ and honest nodes send are valid.
According to Algorithm~\ref{algo:client_disperse}, $v_k$ sends a \emph{disperse} message to each node. Thus, every honest node receives its row and column data eventually and returns an \emph{ack} message. As the client gets at least $n-t$ \emph{ack} messages, it reliably broadcasts the set $S$ containing information about the $n-t$ nodes from which these messages have been received and further sends \emph{spread} messages to all nodes, causing all honest nodes to send \emph{row\_info}, \emph{column\_info}, and \emph{info\_sent} messages. Hence, all honest nodes eventually receive at least $n-t$ \emph{info\_sent} messages and broadcast \emph{ready} messages.
Since every honest node eventually receives $n-t$ \emph{ready} messages as well as its storage data (i.e., \texttt{get\_storage\_data}$(id) \ne \bot$), and sends \emph{column\_info} messages, every honest node terminates the dispersal.
\end{proof}

If the client is Byzantine, it is possible that not all honest nodes obtain a \emph{disperse} message containing their row data. As the following lemma shows, an honest node terminating the dispersal implies that there must be an honest node that learns that at least $n-t$
nodes claim to have sent their \emph{spread} messages.

\begin{lemma}\label{lemma:info_sent} If an honest node terminates the dispersal, there is an honest node that has received at least $n-t$ valid \emph{info\_sent} messages.
\end{lemma}
\begin{proof}
Let $v$ be an honest node that terminates the dispersal, which means that $v$ must have received $n-t$ valid \emph{ready} messages. Let $v'$ be the first honest node to send a \emph{ready} message. If $v'$ broadcast a \emph{ready} message because it received $t+1$ \emph{ready} messages, it must have received such a message from at least $1$ honest node, contradicting the assumption that $v'$ is the first node to send a \emph{ready} message. It follows that $v'$ has received $n-t$ valid \emph{info\_sent} messages.
\end{proof}

Protocol~$\mathcal{P}$ further satisfies the agreement condition as the following lemma shows.

\begin{lemma}[Agreement]\label{lemma:avid_agreement} If an honest node terminates a dispersal, then all honest nodes eventually terminate the dispersal.
\end{lemma}
\begin{proof}
Assume that honest node $v$ terminated a dispersal for some data item $m$.
It must have received $n-t$ \emph{ready} messages, at least $n-2t$ of these messages must have been broadcast by honest nodes. Hence it follows that every honest node eventually receives at least $n-2t \ge t+1$ \emph{ready} messages and broadcasts a \emph{ready} message itself. As a result, all honest nodes eventually receive at least $n-t$ \emph{ready} messages.
Since $v$ terminated the dispersal, Lemma~\ref{lemma:info_sent} implies that there must be an honest node $v'$ that has received at least $n-t$ valid \emph{info\_sent} messages. At least $n-2t$ \emph{info\_sent} messages must have been broadcast by honest nodes, which implies that there is a set $S$ that is broadcast reliably and will therefore reach all honest nodes eventually. Moreover, it further entails that every node $v_k$, $(v_k, \_) \notin S$, will eventually get at least $2(n-2t)>n-t$ elements of its row $r_k$ and therefore be able to decode $r_k$. As all honest nodes eventually have their row data and set $S$, every honest node sends a \emph{column\_info} message to each node that does not occur in $S$. Thus, every honest node receives at least $n-t$ elements of its column and can decode it. Since every honest node receives $n-t$ ready messages, obtains its row and column data, and sends \emph{column\_info} messages, it terminates the dispersal.
\end{proof}

Next, we show that data items can be retrieved when dispersed to $\ell = n-t$ honest nodes and that the correctness property holds.

\begin{lemma}[Availability]\label{lemma:avid_availability}
If $\ell$ honest nodes terminate the dispersal of $m$, then any client can eventually retrieve some data item.
\end{lemma}
\begin{proof}
If $\ell = n-t$ honest nodes terminate the dispersal, they all store their row data. Thus, after a client broadcasts a \emph{retrieve} request for the corresponding identifier, it will eventually receive at least $n-t$ responses containing valid row data and return some data, which may be the empty data item $\bot$.
\end{proof}

\begin{lemma}[Correctness]\label{lemma:avid_correctness}
If $\ell$ honest nodes terminate the dispersal of $m$ and clients $c'$ and $c''$ retrieve $m'$ and $m''$, then $m' = m''$. If the writing client was honest, then $m = m' = m''$.
\end{lemma}
\begin{proof}

Let $m$ be any data item for which $\ell = n-t$ honest nodes have their storage data. Due to Lemma~\ref{lemma:avid_availability}, any client $c$ retrieves some data item eventually. More specifically, it obtains $n-t$ rows that can be utilized to construct matrix $A$ containing $(n-t)^2$ data parts.
According to Algorithm~\ref{algo:client_retrieve}, client $c$ reconstructs the entire matrix $M$ and the commitment $\mathcal{C}_M$.

If the writing client is honest, matrix $M$ and the commitment $\mathcal{C}_M$ are constructed correctly for data item $m$, implying that $\mathcal{C}_M = id$ at all clients, which proceed to return the same data item $m$ extracted from $A$.
On the other hand, if the writing client is Byzantine, it is possible that the identifier $id$ does not match the commitment $\mathcal{C}_M$. However, in this case, $\mathcal{C}_M \ne id$ must hold for every client under the assumption that it is infeasible for a Byzantine client to violate the binding property of the commitment scheme, resulting in $\bot$ being returned at all clients, which proves the claim.
\end{proof}

It remains to show that the recovery mechanism works correctly. When replacing a (failed) node $v_i$ with some replacement node $v_i'$, the storage nodes must ensure that $v_i'$ eventually receives its storage data for any data item $m$ that has been dispersed to $n-t$ honest nodes (regardless of whether $v_i$ obtained this data).

\begin{lemma}[Recovery] When running the recovery algorithm of $\mathcal{P}$ for any data item $m$ dispersed to $n-t$ honest nodes, the replacement node $v_i'$ eventually stores $(r_i,c_i,\pi_i^r,\pi_i^c)$.
\end{lemma}\label{lemma:recovery}
\begin{proof}
Consider any node $v_k$ among the $n-t$ honest nodes that have their storage data for data item $m$.
When Algorithm~\ref{algo:recover} is executed at $v_k$ for message $m$ with identifier $id$, it encodes $r_k$ and $c_k$ into a row $r_k'$ and column $c_k'$ of $n$ elements, respectively, and generates the corresponding proofs ${\pi_k^r}'$ and ${\pi_k^c}'$. Subsequently, $v_k$ sends a \emph{recover} message to $v_i'$ containing $c_k'[i]$, $r_k'[i]$, and the corresponding proofs ${\pi_k^r}'[i]$ and ${\pi_k^c}'[i]$.

Thus, node $v_i'$ eventually receives at least $n-t$ valid \emph{recover} messages from distinct nodes.
By decoding the row and the column elements, $v_i'$ obtains $r_i$ and $c_i$.
If $v_i'$ did not receive a message from some node $v_j$, $j\in[n-t]$, the proofs $\pi_M[i, j]$ and $\pi_M[j, i]$ will be missing. However, since $v_i'$ has row $r_i$, column $c_i$, and $n-t$ proofs for both its row and column, it can create all proofs for $r_i$ and $c_i$ using the function \texttt{gen\_proofs} as discussed in \S\ref{sec:model}.
Hence, $v_i'$ can reconstruct and store $(r_i, c_i, \pi_i^r, \pi_i^c)$ as required.
\end{proof}

\subsubsection{Complexity}

The space complexity of Protocol $\mathcal{P}$ is the sum of all bytes that the nodes store for any data item $m$ of size $|m|$.

\begin{theorem}[Space complexity]
\label{theorem:space_complexity}
The space complexity of Protocol $\mathcal{P}$ for storing $m$ is $3|m| + O(n^2|\pi|)$, where $|\pi|$ denotes the size of a commitment proof $\pi$.
\end{theorem}
\begin{proof}
The matrix $M$ contains $m$ in $(n-t)(n-t)$ elements, i.e., the size of a single element is ${|m|}/{(n-t)^2}$.
Each node $v_i$ stores $r_i$ and $c_i$ containing $n-t$ elements each. Thus, all $n$ nodes together store $$2 \cdot n \cdot (n-t)\cdot \frac{|m|}{(n-t)^2} < 3|m|,$$ where the inequality follows from the bound $t < n/3$. Since each node stores the proof of commitment for its $2(n-t)$ elements of $M$, all nodes together store $O(n^2)$ proofs of size $|\pi|$.
\end{proof}

As far as the communication complexity is concerned, Protocol $\mathcal{P}$ guarantees the following bounds for the different operations.

\begin{theorem}[Communication complexity]
The communication complexity of Protocol $\mathcal{P}$ for $m$ is
%\begin{itemize}
 %\item
 $6|m| + O(n^2|\pi|)$ for dispersal,
 %\item
 $\frac{3}{2}|m| + O(n^2|\pi|)$ for retrieval, and
 %\item
 $\frac{9}{2}\frac{|m|}{n} + O(n|\pi|)$ for recovery.
%\end{itemize}
\end{theorem}
\begin{proof}
\emph{Dispersal}. The client first sends the row and column data to all nodes directly in \emph{disperse} messages, which together require $3|m|+O(n^2|\pi|)$ bits to be sent due to Theorem~\ref{theorem:space_complexity}. The reliable broadcast of $S$ of size $|S| \in O(n|\kappa|)$, where $\kappa$ denotes the size of a signature, has a communication complexity of $O(n^2|\kappa|)$ assuming the use of an efficient reliable broadcast protocol. We further assume that $|\kappa| \in O(|\pi|)$.

Due to $t < n/3$, every element $M[i,j]$ has a size of $$\frac{|m|}{(n-t)^2} < \frac{9}{4}\frac{|m|}{n^2}.$$
The client sends a \emph{spread} message containing $t$ elements, including a proof, to each node. When a node receives a 
\emph{spread} message, it sends two elements and the corresponding proofs, plus the commitment, to each of the $t$ nodes that do not occur in $S$. These two rounds together incur a communication complexity of
$$(1 + 2)\cdot n \cdot t \cdot \left(\frac{9}{4}\frac{|m|}{n^2} + O(|\pi|)\right) < \frac{9}{4}|m| + O(n^2|\pi|).$$

Moreover, each node sends a row element of size $\frac{|m|}{(n-t)^2} < \frac{9}{4}\frac{|m|}{n^2}$ to the $t$ nodes that do not occur in $S$ in a \emph{column\_info} message, which adds
$$n \cdot t \cdot \left(\frac{9}{4}\frac{|m|}{n^2} + O(|\pi|)\right) < \frac{3}{4} |m| + O(n^2|\pi|)$$
to the communication complexity.
Lastly, the constant-sized \emph{ack} and \emph{info\_sent} messages require $O(n)$ and $O(n^2)$ bits to be sent, respectively. Thus, the communication complexity of the entire dispersal process is upper bounded by $6|m| + O(n^2|\pi|)$.
\\
\emph{Retrieval}. Each node sends its stored row data to the client, together with $n-t$ proofs and the commitment, therefore the information sent around amounts to 
\begin{align*}
n\left(\frac{|m|}{(n-t)^2}(n-t) + O(n|\pi|)\right) &= \frac{n|m|}{(n-t)} + O(n^2|\pi|) \\
&< \frac{3}{2}|m| + O(n^2|\pi|),
\end{align*}
again using that $t < n/3$.
\\
\emph{Recovery}. Every node sends two elements of $M$ together with the corresponding proofs and the commitment to the replacement node for a total size of
$$n \cdot 2 \cdot \left(\frac{|m|}{(n-t)^2} + O(|\pi|)\right) < \frac{9}{2}\frac{|m|}{n} + O(n|\pi|).$$
\end{proof}

Note that the additive term for the commitments and their proofs for the space complexity and the communication complexity of dispersal and retrieval can be reduced by a factor of $n$ using batching techniques (see e.g.,~\cite{alhaddad2024}).
In this case, the additive term becomes $O(n|\pi| + n^2\log(n))$, where the term independent of $|\pi|$ is due to the fact that $O(n^2)$ elements of size at least $\lceil\log(n)\rceil$ bits---the minimum element size as discussed in \S\ref{sec:model}---are both stored and transmitted.

We conclude this section with a round complexity analysis.

\begin{theorem}[Round complexity]
 For data item $m$, Protocol $\mathcal{P}$ requires 7 communication rounds for dispersal, 2 rounds for retrieval, and 1 round for recovery.
\end{theorem}

\begin{proof}
It is easy to see that the retrieval of a confirmed data item requires $2$ rounds of communication: The client sends a \emph{retrieve} message to all nodes and waits for the corresponding \emph{response} messages. The recovery algorithm consists of a single \emph{recover} message sent to the replacement node. The dispersal algorithm, on the other hand, is more complex. We focus on the complexity when the client is honest because the dispersal may not terminate when the writing client is Byzantine. As shown in Figure~\ref{fig:dispersal}, there are $5$ steps; however, a reliable broadcast requires more than one round of communication. Efficient reliable broadcast is possible in $3$ rounds~\cite{cachin2005, locher2024} without additional cryptographic assumptions. Therefore, the nodes receive set $S$ after $5$ rounds and then send \emph{row\_info} and \emph{info\_sent} messages in the sixth round, followed by the \emph{col\_info} and \emph{ready} messages in the seventh round.
\end{proof}

If latency is a concern, the number of rounds can be reduced through a few modifications: The set $S$ can be added to the \emph{spread} payload, which enables the nodes to send \emph{row\_info} and \emph{info\_sent} messages already in the fourth round, resulting in a total of $6$ rounds without affecting the communication complexity. It is easy to see that a malicious client sending different sets $S$ does not affect correctness.
The number of rounds can further be reduced to $5$ in two ways. First, set $S$ can be added to the \emph{row\_info} messages too, which obviates the use of reliable broadcast but introduces a cost of $O(n^3|\pi|)$ as every node broadcasts set $S$ of size $O(n|\pi|)$. Second, this cubic term can be avoided by using a two-round reliable broadcast algorithm~\cite{abraham2021, locher2024}. However, this algorithm relies on threshold cryptography, introducing another security assumption and requiring a secure setup of threshold key shares.

% !TEX root = ./avid_spaa.tex
\section{Protocol Variant}
\label{sec:avid_msr}

The protocol of the previous section can be modified to reduce the storage overhead at the expense of retrieval and recovery performance. We parameterize this trade-off by a storage overhead factor $\gamma \coloneqq n/\ell$, where $\ell$ is the number of nodes that must provide their data for the retrieval of $m$. As discussed below, our protocol variant covers the range $\gamma \in [\frac{3}{2}, n]$, with smaller $\gamma$ requiring higher retrieval and recovery costs.
For reference, the main protocol of \S\ref{sec:avid} achieves a storage overhead factor of $3$.

The main idea is to use a simplified version of our dispersal algorithm introduced in \S\ref{sec:dispersal} together with a 
\emph{minimum storage regenerating (MSR) code}~\cite{dimakis2010}, which inherently supports efficient recovery. As the MSR structure itself encodes sufficient redundancy for repair, nodes only need to store row data and the corresponding proofs. For the sake of brevity, we provide a high-level description without a formal specification and analysis.
Since this protocol variant is based on the same techniques, its correctness can be verified using similar arguments. We further restrict the discussion to the complexity with respect to the message size $|m|$, without considering the total size of all proof data and control messages, which are again bounded by $O(n^2|\pi|)$.

While an MSR code also requires a threshold $\ell$ of elements sufficient for reconstruction, it further has a parameter $d$, which specifies the number of nodes from which data needs to be obtained to recover the state of a node.
In a Byzantine setting, it must hold that $\ell \le n-t$ and $d \le n-t$ because $t$ nodes may not respond.
For $t \coloneqq \lfloor(n-1)/3\rfloor$, we have that $\ell \le \frac{2}{3}n$ as $n \rightarrow \infty$, which together with the trivial lower bound $\ell \ge 1$ implies the claimed range $\gamma = n/\ell \in [\frac{3}{2}, n]$.
For node recovery, a replacement node needs to obtain ${|m|}/(\ell(d-\ell+1))$ bits from $d$ distinct nodes.

Since our protocol is based on a matrix embedding, we apply the same concept with an MSR code, placing $m$ into an $\ell \times \ell$-matrix, which results in an element size of $|m|/\ell^2$. Each node stores row data consisting of $\ell$ elements for a total of $|m|/\ell$ bits. As all nodes store the same amount of data, we get that the space complexity is $n|m|/\ell$. Recall that $\gamma \coloneqq n/\ell$, so the space complexity is $\gamma|m|$, which shows that $\gamma$ indeed corresponds to the overhead factor.

The modified dispersal algorithm works as follows. The client $c$ first sends only the row data to the nodes and collects $n-t$ \emph{ack} messages. As in Protocol $\mathcal{P}$, $c$ then reliably broadcasts the set $S$ containing $n-t$ valid node-signature pairs but, unlike Protocol $\mathcal{P}$, $c$ sends $2t$ elements, namely $M[j,i]$ and $M[j,i+n]$ together with the corresponding proofs for all $v_j$ that do not occur in $S$, to each node $v_i$ in a \emph{spread} message. For each such index $j$, $v_i$ forwards $M[j,i]$ and $M[j,i+n]$ to node $v_j$ in a \emph{row\_info} message. Every node still broadcasts an \emph{info\_sent} message after disseminating the \emph{spread} messages as before. In fact, the rest of the algorithm remains the same, except that \emph{column\_info} messages are not required.
The retrieval algorithm remains effectively unchanged, with the only difference that rows are requested from $\ell + t \le n$ nodes, taking into account that up to $t$ nodes may not respond. Similarly, recovery data must be requested from $d+t \le n$ nodes as mentioned before.

It remains to derive the communication complexity of the different operations with respect to parameter $\gamma$.
As far as dispersal is concerned, the client first transmits all data to be stored, i.e., $\gamma|m|$ bits. The $2t$ elements in each \emph{spread} message are first sent from the client to all nodes, followed by every node sending $2$ of these elements to each of the $t$ nodes that do not occur in $S$, which amounts to
$$(2t \cdot n + n \cdot 2t)\frac{|m|}{\ell^2} = \frac{4t}{n}\gamma^2|m| < \frac{4}{3}\gamma^2|m|$$ since $t < n/3$.
Thus, the communication complexity of dispersal is $$\left(1 + \tfrac{4}{3}\gamma\right)\gamma|m|.$$

As far as retrieval is concerned, since row data is requested from $\ell + t$ nodes, the communication complexity is
$$\frac{|m|}{\ell}(\ell+t) = \left(1 + \frac{t}{\ell}\right)|m| < \left(1+\frac{\gamma}{3}\right)|m|,$$ again due to the fact that $t < n/3$.

Lastly, the communication complexity for recovery is $$\frac{|m|(d+t)}{\ell(d-\ell+1)},$$ which is minimized when $d \le n-t$ is maximized, i.e., $d = n-t$.
If further $d = \ell$, then we have that
$$\frac{|m|(d+t)}{\ell(d-\ell+1)} = \frac{|m|n}{\ell} = \gamma|m|,$$
i.e., the communication complexity of recovery matches the communication complexity of retrieval in this extreme case, which corresponds to $\gamma = \frac{3}{2}$ because $\ell = \frac{2}{3}n$ for $t \coloneqq \lfloor(n-1)/3\rfloor$ and $n \rightarrow \infty$.

For $\ell < d = n-t$, the communication complexity is bounded by
 $$\frac{|m|(d+t)}{\ell(d-\ell+1)} < \frac{|m|n}{\ell(n-t-\ell)}  \stackrel{t < n/3}{<} \frac{\gamma^2}{\frac{2}{3}\gamma - 1}\frac{|m|}{n}.$$
 
Table~\ref{tab:comparison} shows concrete space and communication complexities for $\gamma = 1.8$ and $\gamma =  1.5$. 
As efficient recovery is not possible for $\gamma =  1.5$, it is worth pointing out that any erasure code
with the properties specified in \S\ref{sec:model}
 can be used instead of an MSR code. In this case, the only difference to protocol $\mathcal{P}$ in \S\ref{sec:avid} is that the nodes neither store nor distribute any column data.
As stated before, the numbers in the table reveal that this protocol variant may be preferable when a low space complexity and communication complexity of dispersal is prioritized for the considered use case.

% !TEX root = ./avid_spaa.tex
\section{Related Work}
\label{sec:related_work}

%\enlargethispage{3pt}

Rabin~\cite{rabin1990} established the foundation of verifiable information dispersal, providing a primitive for distributing data with verifiable consistency.
Cachin and Tessaro~\cite{cachin2005} subsequently generalized this concept to asynchronous networks, achieving optimal storage and retrieval complexity.
Since then, AVID has become a core component in various systems, including state-machine replication~\cite{das2021,hendricks2007}, secure multiparty computation~\cite{lu2019}, decentralized mass storage~\cite{damgaard2019,fisch2019,juels2007,goren2025} and distributed ledgers~\cite{cecchetti2019,nazirkhanova2022,yang2022}.

\paragraph*{Storage and Dispersal Efficiency.} 
The original AVID protocol~\cite{cachin2005} achieves optimal storage complexity but suffers from a communication complexity for dispersal of $O(n|m|)$ because each storage node obtains the entire message but retains only its fragment to support later retrievals.
Subsequent work~\cite{alhaddad2021,das2021,hendricks2007,yang2022} proposed protocols that offer varying trade-offs between fault tolerance and resource consumption, particularly storage and communication overhead.
Most of these protocols follow the same pattern: First, the data is erasure-coded into multiple fragments, then a fragment is sent to each node, followed by reliably broadcasting a commitment to the set of all fragments so each node can verify that its own fragment is consistent with the commitment. This approach guarantees that each node stores data proportional to the size of its own fragment and achieves near-optimal storage, dispersal, and retrieval costs.

Some of these works have focused on reducing the dispersal cost to $O(|m|)$.
Alhaddad et al.~\cite{alhaddad2022} achieve a dispersal, storage, and retrieval cost of $3|m|$ using only collision-resistant hash functions.
In follow-up work, Alhaddad et al.~\cite{alhaddad2024} attained optimal cost for storage and retrieval but at an increased communication complexity of $9|m|$ for dispersal by utilizing a two-dimensional encoding scheme.
In their construction, the message is mapped into a matrix where rows and columns are individually encoded with maximum distance separable (MDS) codes~\cite{singleton1964}. To recover a column, $n-t$ elements are necessary, while $n-2t$ elements suffice to recover a row. 
During the dispersal algorithm, the client sends data from an encoded column to each node, and each node then forwards an element to the other nodes and reconstructs its row data from the received elements. Together with a two-dimensional commitment scheme, this ensures that consistency can be verified.
Only the row data is stored by the nodes, whereas the column data is merely used during the dispersal algorithm and deleted afterwards.
They also study how to cope with fail-stop nodes (in addition to the Byzantine nodes) and 
support batching and partial data retrieval for the scenario where a reader is only interested in parts of the data. %, they do not need to gather the whole data.
Some of their constructions are \emph{robust} in the sense that, upon successful dispersal, the fragments are consistent with the encoding of the message that can be subsequently retrieved.
All these properties can be supported by our protocols as well. However, like most classic AVID protocols, they lack efficient node recovery mechanisms; replacing a failed node requires downloading the entire data item, and the constant in the dispersal communication complexity is larger than in our protocols.

\paragraph*{Repair and Node Recovery.} 
The problem of efficient node repair in distributed storage was popularized by Dimakis et al.~\cite{dimakis2010} through the introduction of \emph{regenerating codes}. 
These codes, particularly minimum storage regenerating (MSR) codes, enable the reconstruction of a missing node's data share by downloading significantly less than the original data size at the expense of a slightly higher storage and retrieval complexity.

The Red Stuff protocol, introduced by Danezis et al.~\cite{danezis2025}, achieves a recovery bandwidth that is inversely proportional to the number of nodes. At a high level, Red Stuff uses an \emph{asymmetric} two-dimensional encoding: the data is placed in an $n\times n$ matrix whose rows require $n-t$ elements to decode and whose columns require $n-2t$ elements to decode. Each node $v_i$ stores the first $n-t$ elements of row $i$ and the first $n-2t$ elements of column $i$.
Each matrix element has size $|m|/((n-t)(n-2t)) < 4.5|m|/n^2$ for $t < n/3$, resulting in a storage complexity of $n(n-t+n-2t)\cdot 4.5|m|/n^2 \approx 4.5|m|$, ignoring commitment proofs.
This scheme provides the redundancy necessary for efficient recovery, but the storage overhead factor is rather large.
During recovery of node $v_i$, up to $2t+1$ nodes send the $i^{th}$ element of their row and up to $n$ nodes send their $i^{th}$ column element, for a total of $(n+2t+1)\cdot 4.5|m|/n^2 \approx 7.5|m|/n$ for $t \approx n/3$.
A dispersal requires $4.5|m|$ bits, but up to $t$ honest nodes may not have received their data at the time of the first retrieval and must run recovery first, rendering the effective dispersal complexity $\approx 4.5|m|+ t\cdot7.5|m|/n \approx 7|m|$ for $t\approx n/3$.

Our work bridges the gap between storage-optimal but repair-heavy protocols such as the protocol by Alhaddad et al.~\cite{alhaddad2024} and repair-optimal but storage-heavy protocols like Red Stuff~\cite{danezis2025}.
The performance gains stem from two co-designed mechanisms.

\textbf{Encoding.}
Red Stuff uses an $n\times n$ matrix with an asymmetric encoding ($(n,n-t)$-row-wise, $(n,n-2t)$-column-wise), storing $n-t$ row elements and $n-2t$ column elements per node, each of size $4.5|m|/n^2$ for $t<n/3$, resulting in $4.5|m|/n$ bits per node.
We use an $n\times 2n$ matrix with a symmetric $(n,n-t)$-column-wise encoding followed by a $(2n,n-t)$-row-wise encoding; the element size is $2.25|m|/n^2$ and each node stores $2(n-t)\approx 4n/3$ elements, yielding $3|m|/n$ bits per node.

\textbf{Dispersal.}
Naively disseminating the $n\times 2n$ matrix would absorb the encoding gains.
The savings arise because after the client sends row and column data and collects $n-t$ signed acknowledgments in the set $S$, every node that has contributed an acknowledgment to $S$ already holds one element of every other node's row (from its own column) and one element of every other node's column (from its own row).
Therefore, the client only needs to send the second half of the rows for the $t$ nodes that do not occur in $S$, reducing the dispersal cost to $6|m|$ versus $7|m|$ for Red Stuff and $9|m|$ for the protocol by Alhaddad et al.~\cite{alhaddad2024}.
%Furthermore, our constructions support efficient proof-of-storage auditing. The proposed mechanism in \S\ref{sec:proof_of_storage} refines the sketch in the Red Stuff paper.

% !TEX root = ./avid_spaa.tex
\section{Conclusion}
\label{sec:conclusion}

There is a growing demand for distributed information dispersal and retrieval solutions for a broad range of use cases with different requirements. 
In systems with high churn rates or large system states, it is of paramount importance to bound the cost of dispersals and state recovery, not just storage overhead. Previous work showed that the state of nodes can be recovered efficiently---at the cost of a noticeably larger space complexity and communication complexity for other operations. In this work, we showed that it is possible to improve dispersal and recovery operations, while achieving optimal communication complexity for retrieval.
Determining the lower bounds for the trade-offs across all operations simultaneously is an intriguing open problem that may be addressed in future research.

\bibliographystyle{ACM-Reference-Format}
\balance
\bibliography{decentralized_storage}

\end{document}